\documentclass[11pt]{article}

\pdfoutput=1

\usepackage{style}

\title{Degree Balance as a Fine-Grained Complexity Boundary for Quantum SAT}

\author[1]{Atsuya Hasegawa\thanks{atsuya.hasegawa@math.nagoya-u.ac.jp}}
\author[2]{Jonas Kamminga\thanks{jonas.kamminga@upb.de}}
\author[1]{Fran{\c{c}}ois Le Gall\thanks{legall@math.nagoya-u.ac.jp}}
\author[3]{Suguru Tamaki\thanks{tamak@sis.u-hyogo.ac.jp}}

\affil[1]{\textit{Graduate School of Mathematics, Nagoya University, Japan}}
\affil[2]{\textit{Department of Computer Science and Institute for Photonic Quantum Systems (PhoQS), Paderborn University, Germany}}
\affil[3]{\textit{{Graduate School of Information Science, University of Hyogo, Japan}}}

\date{}

\begin{document}

\maketitle

\begin{abstract}
The local Hamiltonian problem is the canonical $\mathsf{QMA}$-complete problem, and $O(2^n)$ time classical algorithms and $O(2^{n/2})$ time quantum algorithms are known to solve the problem in the worst case. It is not clear how to improve these brute force strategies for a broad class of the problem because ground states are highly entangled in general, and we cannot directly apply known strategies for classical CSPs.

In this work, we present exponentially faster classical and quantum algorithms under two mild assumptions: (1) the Hamiltonian is frustration-free on YES instances, and (2) it is approximately regular, meaning that every qubit is acted upon by approximately the same number of constraints.

We complement these upper bounds by showing that, assuming (Q)SETH, quantum 5-SAT admits no non-trivial worst-case speedup. Our lower bound further demonstrates that the dependence of our algorithms on regularity is in some sense nearly optimal. 
Specifically, quantum 5-SAT remains (Q)SETH-hard even for Hamiltonians in which all but $O(\sqrt{n})$ qubits participate in only constantly many constraints, while the remaining $O(\sqrt{n})$ qubits each participate in $O(\sqrt{n})$ constraints. 
By contrast, if either the size of this high-degree subset or the degrees of its qubits is reduced by a factor of $n^\delta$, for any $\delta>0$, our algorithm solves the problem in time $O(2^{(1- \varepsilon)n})$ for some $\varepsilon>0$.

Together, our upper and lower bounds establish a fine-grained complexity dichotomy for quantum satisfiability.
\end{abstract}

\clearpage

\tableofcontents

\clearpage

\section{Introduction}
\subsection{Background}
\paragraph{\texorpdfstring{$k$-SAT versus MAX-$k$-SAT}{k-SAT versus MAX-k-SAT} in fine-grained complexity.}
A central question in fine-grained complexity is whether exhaustive search over $2^n$ assignments can be improved by a constant in the exponent. Classical constraint satisfaction exhibits a striking distinction between \emph{feasibility} and \emph{exact optimization} in this respect.

For every fixed $k$, $k$-SAT admits an algorithm running in
\[
    O^*(2^{(1-\varepsilon_k)n})
\]
for some constant $\varepsilon_k>0$ depending only on $k$; this follows, for example, from Sch\"oning's algorithm and the PPSZ line of work \cite{SchoningFOCS99,Paturi+JACM05}. Thus NP-completeness does not prevent a constant improvement over exhaustive search for fixed-width SAT.

The situation is much less understood for exact optimization.  MAX-$2$-SAT admits exponential-time algorithms beating $2^n$ \cite{WilliamsTCS05}, but already for exact MAX-E$3$-SAT, where every clause contains exactly 3 variables, it is a longstanding open problem whether there is any constant $\varepsilon>0$ for which the problem can be solved in $O^*(2^{(1-\varepsilon)n})$ time (see, e.g., \cite{gall2026dequantizing} for a recent discussion). At the exponential-time scale, satisfiability therefore appears capable of supporting pruning and structural reductions that exact optimization may not.

\paragraph{Quantum $k$-SAT and quantum MAX-$k$-SAT}
Kitaev \cite{Kitaev+02} introduced the $k$-local Hamiltonian problem ($k$-LH), which asks to estimate the ground-state energy of a $k$-local Hamiltonian 
\[
	H=\sum_{i=1}^m h_i\,,
\]
acting on $n$ qubits, where $m=\poly(n)$ and $||H||\le \poly(n)$, with $1/\poly(n)$ additive error.
Kitaev also showed that the problem is $\QMA$-complete, i.e., any efficient quantum verification procedure can be embedded into the local Hamiltonian problem. 
The local Hamiltonian problem lies at the interface between physics and computer science, and has been the central problem in quantum complexity theory.  

The quantum satisfiability problem (QSAT), introduced by Bravyi \cite{bravyi06}, is similar to the local Hamiltonian problem, except it asks whether the ground state energy is zero. That is, it asks whether all constraints $h_i$ can be simultaneously satisfied and is thus a feasibility problem as opposed to the local Hamiltonian problem which is an optimization problem.

The $k$-local Hamiltonian problem (and thus QSAT) can be solved classically in $\myO{2^n}$\footnote{In this paper the notation $\myO{\cdot}$ suppresses $\poly(n)$ factors.} time using the power method or its variant the Lanczos method \cite{KW92,Lanczos50}, for $k=O(\log n)$.
On a quantum computer, we can use Grover search \cite{GroverSTOC96} combined with other techniques to achieve running time $\myO{2^{n/2}}$ \cite{Apeldoorn+20,Ge+19,Gilyen+STOC19,Kerzner+24,Lin+20,Martyn+21,Poulin+09}.

The SAT/MAX-SAT gap suggests a natural analogous question for quantum constraint problems.

\paragraph{Starting point: $k$-LH is SETH hard.}
The starting point of the work is a recent result by Chia, Hasegawa, Le Gall, and Shen \cite{ChiaHLGY26} showing that the running time of the known exponential-time algorithms is essentially optimal already for $3$-local Hamiltonian. Assuming SETH, there is no classical $O(2^{(1-\varepsilon)n})$-time algorithm for any constant $\varepsilon>0$; assuming QSETH, there is no quantum $O(2^{(1-\varepsilon)n/2})$-time algorithm.
Their key technical ingredient is a size-preserving circuit-to-Hamiltonian construction that encodes a $T$-step computation using only sublinear ($o(T)$) clock (or size) overhead.

Thus, the quantum analogue of the apparently hard MAX-SAT side is now understood in a strong fine-grained sense: even locality $3$ needs a $O(2^n)$ runtime (assuming SETH).  This makes the feasibility side the next natural question.

\paragraph{Main question: what about QSAT?}

For a projector Hamiltonian
\[
    H=\sum_{i=1}^m \Pi_i,
\]
Quantum $k$-SAT asks whether $\lambda_{\min}(H)=0$, equivalently whether there is a state lying in the kernel of every local projector. YES instances are therefore frustration-free: the global optimum is known in advance, and the problem is only to decide whether all local constraints can be satisfied simultaneously. 

Does this frustration-freeness help? That is, can quantum $k$-SAT be solved in time $\myO{2^{(1 - \varepsilon)n}}$ for some $\varepsilon > 0$? This is the main question behind the current work.

The classical analogy gives a concrete reason for optimism. For fixed $k$, there are always algorithms for $k$-SAT beating exhaustive search, even though exact MAX-$3$-SAT is not known to do so.  Moreover, Quantum $2$-SAT is polynomial-time solvable \cite{bravyi06}, whereas Quantum $3$-SAT is already $\QMA_1$-complete \cite{GossetN16}\footnote{Recent work \cite{grewal2026qma,liu2026achieving} showed it is in fact $\QMA$-complete!}. Ordinary complexity therefore does not settle the fine-grained question, just as NP-completeness does not prevent faster exponential-time algorithms for classical $k$-SAT.

Our answer is neither a uniform yes nor a uniform no.  Frustration-freeness creates a genuine opportunity to compress the search space, but whether that opportunity can be exploited depends on how the interaction degrees are distributed.

\subsection{Our results}
\paragraph{Frustration-freeness helps, but only when it can be exploited}
Our first result is a positive one that shows that for a broad class of frustration-free instances, brute force can indeed be beaten. Of crucial importance turns out to be $d_q$, the number of local terms acting non-trivially on qubit $q$, and the parameters
\[
    \mu=\frac{km}{n},\qquad
    D_{\max}=\max_q d_q,\qquad
    D_p=\left(\frac1n\sum_{q=1}^n d_q^p\right)^{1/p}.
\]
Here $\mu$ upper bounds the average degree.

\begin{theorem}[Informal version of \cref{theorem:main result for upper bound}]
For every fixed locality $k$, if a frustration-free $k$-local Hamiltonian satisfies $D_{\max}\le C\mu$, or $D_p\le C\mu$ for some fixed $p>1$ and constant $C$, then there is a constant $\varepsilon>0$ such that the problem can be solved in randomized classical time
\[
    O^*(2^{(1-\varepsilon)n})
\]
and quantum time
\[
    O^*(2^{(1-\varepsilon)n/2}).
\]
The constant $\varepsilon$ depends only on $k$ and the degree-balance parameters.
\end{theorem}

The condition $D_{\max}=O(\mu)$ includes regular instances, and geometrically local Hamiltonians. 
The $D_p$-based condition allows greater degree imbalance.
A particularly transparent consequence concerns a small exceptional set of high-degree qubits.

\begin{corollary}[Informal version of \cref{cor:smallnumberhighdegree}]
\label{cor:informalsmallhigh}
Consider a QSAT Hamiltonian where $O(n^a)$ qubits have degree $O(n^b)$ and all other qubits have constant degree.  If
\[
    a+b<1,
\]
then Quantum $k$-SAT admits a constant-exponent improvement over brute force, classically and quantumly.
\end{corollary}

It is tempting to summarize this as ``regular instances are easier.'' There is some truth to this intuition, but it is too coarse to explain the boundary we find.

Indeed, bounded occurrence already helps for classical SAT.  If every variable occurs in at most a fixed constant $A$ clauses, then the formula has at most $An$ nonempty clauses. The linear-size CNF/circuit-SAT algorithms of Calabro, Impagliazzo, and Paturi~\cite{calabro2009complexity} therefore give, for every fixed \(A\), a constant \(\varepsilon(A)>0\), independent of the clause width, such that the formula can be solved in
\[
    O^*(2^{(1-\varepsilon(A))n})
\]
time.
Thus bounded occurrence is itself an algorithmically favorable regime; it should not be described as retaining the full SETH difficulty of unrestricted SAT.

Our result, however, goes beyond this ``small degree is easier'' phenomenon. The algorithm does not require a fixed degree bound: the average and individual degrees may grow with $n$, even polynomially, as long as the degrees are sufficiently balanced so that a linear family of disjoint constraints survives.  

\paragraph{General QSAT is (Q)SETH hard.}
We contrast our algorithmic result by showing that in general, quantum $k$-SAT is (Q)SETH hard already for $k \ge 5$. It thus behaves more like $k$-LH, which is (Q)SETH hard for $k \ge 3$ (\cite{ChiaHLGY26}) than classical SAT, which allows $\myO{2^{(1-\varepsilon)n}}$ algorithms for all constant $k$.

\begin{theorem}[Informal version of \cref{thm:sethhard}]
Assuming SETH, Quantum $5$-SAT cannot be solved in $O^*(2^{(1-\varepsilon)n})$ time for any constant $\varepsilon>0$.
Assuming QSETH, it cannot be solved in $O^*(2^{(1-\varepsilon)n/2})$ time.
\end{theorem}

In fact, we show that the hardness survives even when all but $O(\sqrt n)$ qubits have constant degree and the remaining $O(\sqrt n)$ qubits have degree only $O(\sqrt n)$ (\cref{cor:hardnesssmallhigh}). Thus a vanishing fraction of exceptional qubits can carry enough total interaction mass to restore the full brute-force hardness.

Combined with \cref{cor:informalsmallhigh}, this gives a(n almost) sharp degree-profile picture.
If the number of exceptional qubits is $O(n^a)$ and their degrees are $O(n^b)$, then our algorithm handles the entire region $a+b<1$, while the boundary point
\[
    a=b=\frac12
\]
already supports (Q)SETH-hardness.

\subsection{Our techniques}

\paragraph*{Upper bounds}

The starting point is the frustration-free condition. If $H=\sum_{i=1}^{m}h_i$ is frustration-free and $\ket{\psi}$ is a ground state, then $h_i\ket{\psi}=0$ for every local term $h_i$. Thus, on the support of each $h_i$, the state $\ket{\psi}$ is entirely supported on $\ker(h_i)$. We exploit this by choosing a large collection of local terms whose supports are pairwise disjoint. For these terms, the corresponding subsystems can be restricted simultaneously to their local kernels. Each such restriction reduces the local dimension, and a sufficiently large collection therefore yields an exponential reduction in the dimension of the global search space. 

In the case that $D_{\text{max}} = O(\mu)$, we can use a greedy coloring argument to obtain such a sufficiently large collection. When $D_p = O(\mu)$, we first discard all constraints acting on any qubit with $d_q \ge \Theta$ and show that for an appropriate choice of $\Theta$ we can again use a greedy coloring argument on the remaining instance to obtain the sufficiently large collection.

To turn this reduction into a correct algorithm, we represent the restricted space as the image of an isometry $V$ and define the compressed Hamiltonian by
\[
H_{\mathrm{eff}}=V^\dagger H V.
\]
On YES instances, every frustration-free ground state lies in the image of $V$, and thus a zero-energy ground state is preserved by the compression. On the other hand, $H_{\mathrm{eff}}$ is simply the restriction of $H$ to the subspace $\operatorname{Im}(V)$. Since minimizing over a smaller set of states cannot decrease the minimum energy, we have
\[
\lambda(H_{\mathrm{eff}})\ge \lambda(H).
\]
Hence the compression preserves the zero ground energy of YES instances while maintaining the promise gap on NO instances. Therefore, it suffices to solve the local Hamiltonian problem for the compressed Hamiltonian $H_{\mathrm{eff}}$.
Since $H_{\mathrm{eff}}$ acts on a Hilbert space of exponentially smaller dimension, applying standard ground-energy estimation algorithms to $H_{\mathrm{eff}}$ yields the desired exponential speedup for the original Hamiltonian $H$.

\paragraph*{Lower bounds}
To establish our lower bound we broadly follow Chia, Hasegawa, Le Gall and Shen's proof of SETH hardness of the local Hamiltonian problem \cite{ChiaHLGY26}. That is, we construct a compressed clock that allows reducing a SAT formula with $n$ variables and $O(n)$ clauses to a quantum SAT instance on $n + o(n)$ qubits. An $\myO{2^{(1-\varepsilon)n}}$ time algorithm for the latter would thus contradict SETH.

The techniques from \cite{ChiaHLGY26} do not all transfer straightforwardly to the quantum SAT case. Their clock construction uses a Hamiltonian for enforcing the clock states that is not frustration free and a transition operator that is not a projector. As a result, these are not legal quantum SAT terms. To circumvent these issues we overhaul their clock. We use a combination of a one-hot clock coupled with a unary clock, and a 4-state clock due to Bravyi \cite{bravyi06}. We will still need more locality, which is why our result holds for $5$-local quantum SAT compared to \cite{ChiaHLGY26} $3$-local Hamiltonians.

To establish that (Q)SETH hardness remains even with only $O(\sqrt{n})$ qubits with degree $O(\sqrt{n})$ we construct a new efficient quantum circuit for verifying SAT. The key parameters are that this circuit consists of only $O(n)$ elementary gates and acts on $O(\sqrt{n})$ counter qubits at most $O(\sqrt{n})$ times. The remaining qubits will be acted on only a constant number of times\footnote{This assumes that every variable occurs a bounded number of times in the SAT formula. This can be ensured using the Sparsification Lemma \cite{Impagliazzo+08}.}. 

Chia et al. \cite{ChiaHLGY26} verify a SAT formula by sequentially checking the clauses and incrementing a counter for every true clause. In their verification circuit, all counter qubits are acted on every time any clause is updated resulting in a gate complexity of $O(n \log^2 n)=\Tilde{O}(n)$. We improve this construction by observing that we do not need to know the number of true clauses, only if all clauses are true or not. Accordingly, we can modify the counter updates by applying it only to the qubits needed to make this distinction and showing a suitable ``rigidity'' property. This leads to the gate complexity of $O(n)$ instead of $O(n \log^2 n)$ (\cref{lem:U_compute_Phi}). We apply the same idea to a verification circuit with two $\sqrt{n}$-bit one-hot counters\footnote{Here we apply the concept of the compressed clock in the circuit-to-Hamiltonian construction in \cite{ChiaHLGY26} for the verification circuit for SAT.}, and obtain the desired verification circuits for SAT (\cref{thm:two-unary-sat-verifier}). 

\subsection{Related work}

The use of a large independent set to obtain an exponential speedup is inspired by the 3-SAT algorithm of Hofmeister, Schöning, Schuler, and Watanabe \cite{Hofmeister+2007}, which exploits a large collection of pairwise variable-disjoint clauses to improve Schöning's randomized local-search algorithm. Our approach applies a related principle in the quantum setting. We select pairwise disjoint local terms and, using frustration-freeness, restrict their subsystems to the corresponding local kernels. A linear-size independent set, therefore, yields an exponential reduction in the dimension of the search space.

Buhrman et al. \cite{Buhrman+PRL25} showed classical and quantum algorithms running in time $\myO{2^{n(1-\Omega(\delta/M))}}$ and $\myO{2^{n(1-\Omega(\delta/M))/2}}$, respectively, where $M:=\sum_{i}\|h_i\|$, to estimate the ground-state energy of any $O(1)$-local Hamiltonian up to an additive error $\delta$. To obtain faster exponential time algorithms, we need $\delta=\Omega(M)=\poly(n)$ additive error; in other words, a rough approximation of the ground state energy. On the other hand, we can achieve an exponential speed-up with our algorithm to estimate the ground state energy with $1/\poly(n)$-additive error, while we need the two assumptions for Hamiltonians.

The role of the degree distribution has an independent and, in our view, illuminating parallel in recent work of Le Gall and Tamaki on dequantizing short-path algorithms for exact MAX-$k$-CSP \cite{gall2026dequantizing}. In the unweighted setting, their discussion uses two instance parameters. The first is the normalized degree irregularity
\[
    \mathcal D
      =
    \frac{n\sum_{q=1}^n d_q^2}{k^2m^2}.
\]
For a perfectly balanced instance $\mathcal D=1$; concentration of incidence mass makes $\mathcal D$ larger. The second parameter, usually denoted $\Delta$, measures the normalized separation between the optimum and the average objective value of a random assignment.  Their results give faster-than-brute-force algorithms in a regime with mild degree irregularity and constant optimum--average separation.

For frustration-free quantum $k$-SAT, the analogous optimum--average separation is automatic: the ground-state energy is zero, whereas the maximally mixed state has energy at least $m/2^k$.
What remains is the degree condition, and even the quantitative statistic is the same. With
\[
    \mu=\frac{km}{n},
    \qquad
    D_2=\left(\frac1n\sum_{q=1}^n d_q^2\right)^{1/2},
\]
we have
\[
    \mathcal D
       =
    \frac{n\sum_q d_q^2}{k^2m^2}
       =
    \left(\frac{D_2}{\mu}\right)^2.
\]
Equivalently, if $\rho_2=\mu/D_2$, then $\mathcal D=\rho_2^{-2}$.
Thus the same normalized second moment of the degree distribution emerges in classical exact optimization and in our frustration-free quantum algorithm.

This connection also points beyond the present paper.  For general, frustrated Local Hamiltonians, degree balance alone should not be expected to suffice. A natural open problem is to identify a quantum analogue of the optimum--average separation condition and determine whether the combination of that condition with degree balance yields faster exact algorithms for general Local Hamiltonian. Our results settle the frustration-free endpoint of this question: there the separation is automatic, degree balance yields a speedup, and sufficiently concentrated degree mass can restore SETH-hardness.

\subsection{Open questions}

We list some open questions left by our work.

\begin{itemize}
    \item In this work, we show faster exponential-time algorithms for quantum SAT with regular or geometrically local interactions. Can we obtain a faster exponential-time algorithm to solve the non-frustration-free local Hamiltonian problem whose interactions are geometrically local? Some such local Hamiltonian problems are known to be $\QMA$-hard \cite{Oliveira+08,Biamonte+08,Cubitt+16}.
    \item In this paper, we show the (Q)SETH-hardness of quantum $5$-SAT. Can we improve the locality of $5$ to $4$ and $3$? Note that quantum $2$-SAT is in $\PLang$ \cite{bravyi06}.
    \item What is the complexity of quantum 3-SAT where all $3$-local constraints overlap in a single qubit?
    \item \cref{cor:informalsmallhigh} works when $a + b < 1$, whereas we show (Q)SETH hardness for $a = b = \frac{1}{2}$. Can (Q)SETH hardness also be shown for other combination of $a$ and $b$ summing to $1$? By splitting the clock state over 3 or more registers, instead of 2, it might be possible to show (Q)SETH hardness for decreased $a$ at the cost of increasing $b$ and the locality.
\end{itemize}

\section{Preliminaries}

We use the standard notations in quantum computing and complexity theory. We refer to \cite{NC10,Gharibian2015,deWolf19} for standard references. We denote by $\mathcal{B}\left(\mathbb{C}^{2^n}\right)$ the set of pure states of $n$ qubits.

\subsection{Local Hamiltonian problem} 
A Hamiltonian is a Hermitian operator acting on a quantum register.
Let $H$ be a Hamiltonian, we call the eigenvalues of $H$ the energies of the Hamiltonian.
We denote by $\lambda(H)$ the ground-state energy of $H$, and we denote by the corresponding eigenstate(s) the ground state(s) of $H$. 
We use $\|H\|$ to denote the operator norm of $H$, that is, the largest absolute value of the eigenvalues of $H$.

The $k$-local Hamiltonian is defined as follows.
\begin{definition}[$k$-Local Hamiltonian]\label{def:klocalH}
A Hamiltonian $H$ acting on $n$ qubits is $k$-local if $H$ can be written as $H = \sum_{i=1}^m h_i\footnote{We assume that all the elements of local terms are some constants independent in $n$.}$ for $m=\poly(n)$ and for all $i$, the following holds.
\begin{itemize}
    \item $h_i$ is a positive semidefinite operator.\footnote{
The standard definition of the Local Hamiltonian problem allows each local term $h_i$ to be an arbitrary Hermitian operator. However, replacing
\[
h_i \mapsto h_i-\lambda_{\min}(h_i)I
\]
makes $h_i$ positive semidefinite. This transformation shifts the spectrum of the total Hamiltonian by a known additive constant, and the promise thresholds can be adjusted accordingly. Hence restricting attention to positive-semidefinite local terms incurs no loss of generality.
}
    \item $h_i$ non-trivially acts on at most $k$ qubits.
    \item $\|h_i\|\le \poly(n)$.
\end{itemize} 
\end{definition}

We call $k$ the locality of $H$, and define the $k$-local Hamiltonian problem formally.

\begin{definition}[$k$-local Hamiltonian problem]
    \label{def:LHP}
    The local Hamiltonian problem is a decision problem that asks whether the ground-state energy of a $k$-local Hamiltonian is greater or less than given energy thresholds.
    \begin{itemize}
        \item \textbf{Inputs:} a $k$-local Hamiltonian $H$ acting on $n$ qubits where $k=O(1)$, and two energy thresholds $a,b$ satisfying $b-a\ge \frac{1}{\poly(n)}$.
        \item  \textbf{Outputs:}
        \begin{itemize}
            \item YES, if there exists a quantum state $\ket{\psi}\in\mathcal{B}\left(\mathbb{C}^{2^n}\right)$ such that $\bra{\psi}{H}\ket{\psi}\le a$.
            \item NO, if $\bra{\psi}{H}\ket{\psi}\ge b$ for all $\ket{\psi}\in\mathcal{B}\left(\mathbb{C}^{2^n}\right)$.
        \end{itemize}
    \end{itemize}
\end{definition}

The 2-local Hamiltonian problem is known to be $\QMA$-hard \cite{Kempe+06}, and several exponential-time classical and quantum algorithms are known for the problem \cite{KW92,Kerzner+24}, which can be summarized as follows. 

\begin{fact}\label{fact:exponential-time algo}
Let $H$ be a Hermitian matrix acting on a $D$-dimensional Hilbert
space for $D \leq 2^n$. Suppose $H$ is $\poly(n)$-sparse and $H$ is given in the standard sparse-access model. Suppose further that $\|H\| \le \poly(n)$.

Then, for every $\varepsilon = \frac{1}{\poly(n)}$,
there exist an $O^*(D)$-time randomized classical algorithm and an
$O^*(\sqrt D)$-time quantum algorithm that output, with high
probability, an estimate $\widetilde{\lambda}$ satisfying
\[
\left|
\widetilde{\lambda}
-
\lambda(H)
\right|
\le
\varepsilon.
\]
\end{fact}

In this paper, we focus on frustration-free Hamiltonians.

\begin{definition}[Frustration-free Hamiltonian]
Let $H=\sum_{i=1}^{m} h_i$ be a $k$-local Hamiltonian. We say that $H$ is \emph{frustration-free} if there exists a quantum state
$\ket{\psi}\in\mathcal{B}\left(\mathbb{C}^{2^n}\right)$ such that
\[
h_i\ket{\psi}=0
\]
for all $i\in[m]$.
\end{definition}

For positive-semidefinite local terms,
$H \text{ is frustration-free}$ if and only if $\lambda(H)=0$.
We define the frustration-free local Hamiltonian problem as the special case of the $k$-local Hamiltonian problem
with $a=0$.

We next define the quantum $k$-SAT problem.

\begin{definition}[Quantum $k$-SAT Hamiltonian]
A quantum $k$-SAT Hamiltonian is a $k$-local Hamiltonian
\[
H=\sum_{i=1}^{m}\Pi_i
\]
where each $\Pi_i$ is a projector (i.e., $\Pi_i^2=\Pi_i$) acting non-trivially on at most $k$ qubits.
\end{definition}

\begin{definition}[Quantum $k$-SAT Problem]
\label{def:qksat}
Given a quantum $k$-SAT Hamiltonian $H=\sum_{i=1}^{m}\Pi_i$,
and a threshold $b\ge \frac{1}{\poly(n)}$, decide which of the following holds.
\begin{itemize}
    \item \textbf{YES:} $H$ is frustration-free.
    \item \textbf{NO:} $\lambda(H)\ge b$.
\end{itemize}
\end{definition}

Since every projector is positive semidefinite, every quantum
$k$-SAT Hamiltonian is a $k$-local Hamiltonian with
positive-semidefinite local terms. Quantum $k$-SAT is therefore a special case of the
frustration-free $k$-local Hamiltonian problem.

\subsection{SAT, SETH and QSETH}

A Boolean variable takes a value in ${0,1}$. A Boolean formula is an expression over variables using the logical connectives $\neg$ (NOT), $\vee$ (OR), and $\wedge$ (AND).

Let $\Phi$ be a Boolean formula over variables $x_1,\ldots,x_n$. For an assignment $x\in\{0,1\}^n$, the variable $x_i$ is assigned the $i$th bit of $x$. We write $\Phi(x)$ for the value of $\Phi$ under the assignment $x$.

We define the conjunctive normal form formula as follows.
\begin{definition}[$k$-CNF formula]
A Boolean formula $\Phi$ over variables $x_1,\ldots,x_n$ is a $k$-CNF formula if
\[
    \Phi=\bigwedge_{i=1}^{m} C_i,
\]
where $m=\poly(n)$ and each clause $C_i$ is a disjunction of at most $k$ literals. A literal is either a variable $x_j$ or its negation $\neg x_j$.
\end{definition}

\begin{definition}[$k$-SAT]
Given a $k$-CNF formula $\Phi$, the $k$-SAT problem asks whether there exists an assignment $x\in\{0,1\}^n$ such that $\Phi(x)=1$.
\end{definition}

Though the exact runtime lower bound for $\kSAT$ is still unknown, it is widely believed that the brute-force search is optimal for classical algorithms and the Grover search is optimal for quantum algorithms. The Strong Exponential Time Hypothesis (SETH) \cite{Impagliazzo+01} states that $\kSAT$ needs roughly $2^n$ time for large $k$:
\begin{conjecture}[SETH]\label{conjecture:SETH}
	For all $\varepsilon$, there is some $k\ge 3$ such that $\kSAT$ with $n$ variables and $O(n)$ clauses cannot be solved classically in time $O(2^{(1-\varepsilon)n})$.
\end{conjecture}

A quantum analog was proposed in \cite{Aaronson+CCC20,Buhrman+21} as Quantum Strong Exponential-Time Hypothesis (QSETH).

\begin{conjecture}[Quantum strong exponential time hypothesis (QSETH) \cite{Aaronson+CCC20,Buhrman+21}]\label{conjecture:QSETH}
	For all $\varepsilon$, there is some $k\ge 3$ such that $\kSAT$ with $n$ variables and $O(n)$ clauses cannot be solved quantumly in time $O(2^{(1-\varepsilon)n/2})$.
\end{conjecture}

In the two conjectures above, we can assume the number of clauses to be $O(n)$ via the sparsification lemma \cite{Impagliazzo+08}. 

Lemma 23 in \cite{ChiaHLGY26} showed an efficient quantum circuit to verify $\kSAT$ with the linear number of clauses. The number of gates in the verification circuits was quasi-linear in $n$ ($O(n \log^2 n)$), and we show we can improve it to linear in $n$ by a careful analysis of the gate counts and observing a ``rigidity'' of the counter. See \cref{appendix} for a proof. In this paper, we use the Clifford+$T$ gate set as our elementary gate set, as in Chia et al.~\cite{ChiaHLGY26}.

\begin{lemma}\label{lem:U_compute_Phi}
    For any positive integer $k$, and any $\kCNF$ formula $\Phi$ that contains $n$ variables and $m=O(n)$ clauses, there exists a quantum circuit $U_\Phi$ that acts on $n$ input qubits and $O(\log n)$ ancilla qubits, and $U_\Phi$ consists of $O(n)$ elementary gates such that $U_{\Phi}\ket{x}\reg{in}\otimes\ket{0}\reg{anc}=\ket{\Phi(x)}\reg{out}\otimes\ket{\psi_x}\reg{\overline{out}}$ for all $x\in\{0,1\}^n$, where $\ket{\psi_x}$ is some quantum state depending on $x$.
    The construction of $U_{\Phi}$ can be done in $\poly(n)$ time.

    The circuit $U_\Phi$ will be such that the $i$-th input qubits is acted upon only $O(d_i)$ times, where $d_i$ is the number of clauses acting on the $i$-th variable in $\Phi$. The ancilla qubits will be acted upon $\Tilde{O}(n)$ times. 
\end{lemma}

\section{Faster exponential-time algorithms}

In this section, we present faster exponential-time algorithms for the frustration-free local Hamiltonian problem.
In \cref{subsec:compression}, we first show how to compress a local Hamiltonian into one acting on a smaller-dimensional Hilbert space. In \cref{subsec:construct an independent set}, we then explain how to construct efficiently a heavy family of local terms whose supports are pairwise disjoint. Finally, in \cref{subsec:main result}, we apply the compression to obtain our faster exponential-time algorithms.

Throughout this section, we assume without loss of generality that there are no $0$-local terms, no terms with trivial kernel and no qubits that are only acted on trivially. For each local term $h_i$, we denote by $\ker(h_i)$ and $\rank(h_i)$ the kernel of $h_i$ and rank of the local matrix acting on $\operatorname{supp}(h_i)$.

\subsection{Compression theorem}
\label{subsec:compression}

\begin{theorem}[Compression Theorem]
\label{thm:compression}
Let $H=\sum_{i=1}^{m} h_i$ be a $k$-local Hamiltonian on $n$ qubits, where $k=O(1)$.
Suppose we are given a set
\[
\mathcal S
\subseteq
[m]
\]
of local terms whose supports are pairwise disjoint.

For each $s\in\mathcal S$, define $k_s:=\left|\operatorname{supp}(h_s)\right|$ and $g_s:=\dim\ker(h_s).$
Suppose that $g_s\ge 1$ for every $s\in\mathcal S$, and define
\[
M
:=
\left\{2^{(n-\sum_{s\in\mathcal S}k_s)}\right\}
\prod_{s\in\mathcal S}g_s.
\]

Then there exists a deterministic classical algorithm, running in
polynomial time in $n$ and $m$, that computes

\begin{enumerate}
    \item a succinct description of an isometry
    \[
    V
    :
    \mathbb C^M
    \longrightarrow
    (\mathbb C^2)^{\otimes n},
    \]

    \item an effective Hamiltonian
    \[
    H_{\mathrm{eff}}
    :=
    V^\dagger H V,
    \]
    represented as a sum of efficiently computable local terms.
\end{enumerate}

Furthermore, $\lambda(H_{\mathrm{eff}}) \ge \lambda(H)$, and 
if $\lambda(H)=0$, then $\lambda(H_{\mathrm{eff}})=0$.
\end{theorem}

Since $h_s$ acts on $k_s$ qubits, we have
\[
g_s
=
\dim\ker(h_s)
=
2^{k_s}
-
\operatorname{rank}(h_s).
\]
Hence
\[
M
=
2^n
\prod_{s\in\mathcal S}
\frac{g_s}{2^{k_s}}
=
2^n
\prod_{s\in\mathcal S}
\left(
1-
\frac{\operatorname{rank}(h_s)}{2^{k_s}}
\right).
\]

\begin{proof}
Fix $s\in\mathcal S$. Since $H$ is $k$-local for $k=O(1)$, by the spectral decomposition of $h_s$, we can efficiently compute an orthonormal basis
\[
\left\{
\ket{v_{s,b}}
\right\}_{b=1}^{g_s}
\]
of $\ker(h_s)$.

Define
\[
B_s
:
\mathbb C^{g_s}
\longrightarrow
(\mathbb C^2)^{\otimes k_s}
\]
by
\[
B_s
=
\sum_{b=1}^{g_s}
\ket{v_{s,b}}\!\bra{b},
\]
where
\[
\left\{
\ket{b}
\right\}_{b=1}^{g_s}
\]
is the standard basis of $\mathbb C^{g_s}$. Since $B_s$ maps the standard basis (an
orthonormal basis) to the orthonormal basis
$\{\ket{v_{s,b}}\}_{b=1}^{g_s}$ of $\ker(h_s)$, it is an isometry and
\[
\operatorname{Im}(B_s)
=
\ker(h_s).
\]
Moreover, since every $\ket{v_{s,b}}$ belongs to $\ker(h_s)$,
\[
h_sB_s
=
0.
\]

Since the supports of the terms in $\mathcal S$ are pairwise disjoint,
we define, under the natural tensor-product identification induced by
these supports,
\[
V
=
\left(
\bigotimes_{s\in\mathcal S}B_s
\right)
\otimes
I_{\mathrm{rest}},
\]
where $I_{\mathrm{rest}}$ denotes the identity on the remaining
$n-\sum_{s\in\mathcal S}k_s$ qubits. 
As a tensor product of isometries, $V$ is an isometry.
Furthermore,
\[
\operatorname{Im}(V)
=
\left(
\bigotimes_{s\in\mathcal S}\ker(h_s)
\right)
\otimes
(\mathbb C^2)^{
\otimes
\left(
n-\sum_{s\in\mathcal S}k_s
\right)
},
\]
and its dimension is
\[
\left\{2^{(n-\sum_{s\in\mathcal S}k_s)}\right\}
\prod_{s\in\mathcal S}g_s = M.
\]

We now define
\[
H_{\mathrm{eff}}
:=
V^\dagger H V.
\]
Since $h_sB_s=0$, we have $V^\dagger h_sV=0$ for every $s\in\mathcal S$. 
Therefore,
\[
H_{\mathrm{eff}}
=
\sum_{\eta\notin\mathcal S}
V^\dagger h_\eta V.
\]

We next show that the projected local terms can be computed efficiently.
Fix $\eta\notin\mathcal S$. Since $h_\eta$ acts on at most $k$ qubits
and the supports in $\mathcal S$ are pairwise disjoint, the support of
$h_\eta$ intersects at most $k$ selected supports.

To compute
\[
V^\dagger h_\eta V,
\]
we first extend the local matrix of $h_\eta$ by identities on the
qubits inside the intersected selected supports on which $h_\eta$ does
not act. We then conjugate the resulting operator by the corresponding
isometries $B_s$. This computation involves only constantly many
constant-dimensional matrices and can therefore be performed
efficiently. Hence the local-term description of $H_{\mathrm{eff}}$
can be computed in polynomial time in $n$ and $m$.

It remains to prove the energy claims. For every normalized state
$\ket{\phi}\in\mathbb C^M$, since $V$ is an isometry, $V\ket{\phi}$ is
also normalized. Therefore,
\[
\begin{aligned}
\bra{\phi}H_{\mathrm{eff}}\ket{\phi}
&=
\bra{\phi}V^\dagger HV\ket{\phi}\\
&=
\bra{V\phi}H\ket{V\phi}\\
&\ge
\lambda(H).
\end{aligned}
\]
Taking the minimum over all normalized $\ket{\phi}$ gives
\[
\lambda(H_{\mathrm{eff}})
\ge
\lambda(H).
\]

Now suppose that
\[
\lambda(H)=0.
\]
Let $\ket{\psi}$ be a normalized ground state of $H$. Since every
local term is positive semidefinite,
\[
0
=
\bra{\psi}H\ket{\psi}
=
\sum_{i=1}^{m}
\bra{\psi}h_i\ket{\psi}
\]
implies
\[
\bra{\psi}h_i\ket{\psi}
=
0
\]
for every $i\in[m]$. Since $h_i$ is positive semidefinite, this implies $h_i\ket{\psi}=0$ for every $i\in[m]$.

In particular,
\[
h_s\ket{\psi}
=
0
\]
for every $s\in\mathcal S$. Let
\[
\mathcal S
=
\{s_1,\ldots,s_t\}.
\]
Since the supports of the terms in $\mathcal S$ are pairwise disjoint,
the state $\ket{\psi}$ can be expanded using the orthonormal bases
\[
\left\{
\ket{v_{s,b}}
\right\}_{b=1}^{g_s}
\]
of the spaces $\ker(h_s)$ as
\[
\ket{\psi}
=
\sum_{b_{s_1}=1}^{g_{s_1}}
\cdots
\sum_{b_{s_t}=1}^{g_{s_t}}
\left(
\bigotimes_{s\in\mathcal S}
\ket{v_{s,b_s}}
\right)
\otimes
\ket{\alpha_{(b_s)_{s\in\mathcal S}}},
\]
where each
$\ket{\alpha_{(b_s)_{s\in\mathcal S}}}$ is a vector on the remaining $(n-\sum_{s\in\mathcal S}k_s)$ qubits.\footnote{We use the convention
$(\mathbb C^2)^{\otimes 0}=\mathbb C$. Thus, if the selected supports
cover all qubits, each
$\ket{\alpha_{(b_s)_{s\in\mathcal S}}}$ is simply a complex
coefficient.}

Define
\[
\ket{\phi}
=
\sum_{b_{s_1}=1}^{g_{s_1}}
\cdots
\sum_{b_{s_t}=1}^{g_{s_t}}
\left(
\bigotimes_{s\in\mathcal S}
\ket{b_s}
\right)
\otimes
\ket{\alpha_{(b_s)_{s\in\mathcal S}}}.
\]
It is clear that $V\ket{\phi} = \ket{\psi}$. Since $V$ is an isometry and $\ket{\psi}$ is normalized, $\ket{\phi}$ is also normalized.
Therefore,
\[
\bra{\phi}H_{\mathrm{eff}}\ket{\phi} = \bra{\phi}V^\dagger H V \ket{\phi} = \bra{\psi}H\ket{\psi}=0.
\]
Since $H_{\mathrm{eff}}$ is positive semidefinite, we conclude that
\[
\lambda(H_{\mathrm{eff}})
=
0.
\]
\end{proof}

\subsection{Constructing a weighted disjoint family}
\label{subsec:construct an independent set}

To apply \cref{thm:compression}, we require a collection of local terms whose supports are pairwise disjoint. In the general positive-semidefinite setting, different local terms may yield different amounts of compression. We therefore assign a weight to each local term and seek a heavy disjoint family. 

Let $H=\sum_{i=1}^{m}h_i$ be a $k$-local Hamiltonian. For each $i\in[m]$, define
\[
    k_i:=\left| \operatorname{supp}(h_i) \right|
\]
and
\[
    g_i:=\dim\ker(h_i).
\]

If $g_i=0$ for some $i\in[m]$, then $h_i$ has no zero eigenvalue, and hence $H$
cannot be frustration-free.
In the remainder of this section, we assume that $g_i\ge 1$ for every $i\in[m]$.

For each local term $h_i$, define its weight by
\[
    w_i := \log_2 \frac{2^{k_i}}{g_i}.
\]

By \cref{thm:compression}, selecting a disjoint family $\mathcal S \subseteq [m]$ yields a compressed Hilbert space of dimension
\[
    M = 2^n \prod_{s\in\mathcal S} \frac{g_s}{2^{k_s}}
    = 2^{\,n-\sum_{s\in\mathcal S}w_s}.
\]

Hence, minimizing the compressed dimension is equivalent to maximizing
the total weight $\sum_{s\in\mathcal S}w_s$.

Since $g_i = 2^{k_i} - \operatorname{rank}(h_i)$, we may equivalently write
\[
w_i = -\log_2 \left( 1- \frac{\operatorname{rank}(h_i)}{2^{k_i}} \right).
\]

By the preprocessing above, every local term is nonzero.
Therefore, $\rank(h_i)\ge 1$ for every $i\in[m]$.
Since $k_i\le k$, we obtain
\[
    w_i
    =
    -\log_2\left(
        1-\frac{\rank(h_i)}{2^{k_i}}
    \right)
    \ge
    -\log_2\left(1-2^{-k}\right).
\]
For convenience, define
\[
    \beta_k
    :=
    -\log_2\left(1-2^{-k}\right).
\]
Thus,
\[
    w_i\ge \beta_k
\]
for every $i\in[m]$. Since $k=O(1)$, $\beta_k>0$ is a constant depending only on $k$.

We define the \emph{term-intersection graph}
\[
    G=([m],E)
\]
as follows. Two distinct vertices $i,j\in[m]$ are adjacent if and only
if
\[
    \operatorname{supp}(h_i) \cap \operatorname{supp}(h_j) \neq \emptyset.
\]
A subset
\[
    \mathcal S \subseteq [m]
\]
has pairwise disjoint supports if and only if it is an independent set of $G$.

For each qubit $q\in[n]$, let
\[
d_q
:=
\left|
\left\{
i\in[m]
:
q\in\operatorname{supp}(h_i)
\right\}
\right|
\]
denote the number of local terms acting on $q$. 

Denote the total weight of the graph $G$ as
\begin{align}
    w(G) = \sum_{i = 1}^m w_i
\end{align}
Finding a maximum-weight independent set is $\NP$-hard in general. For our purpose, however, it suffices to find an independent set whose total weight is a sufficiently large fraction of $w(G)$. We now discuss two methods of finding such a large weight set $\calS$.

\paragraph{Maximum-degree construction}

Define
\[
\Dmax
:=
\max_{q\in[n]}d_q.
\]
A bounded-interaction-degree assumption allows us to obtain such an independent set efficiently by a simple greedy coloring argument.

\begin{lemma}[Disjoint Family from Max Degree]
\label{lem:weighted-disjoint-family-max-degree}
There exists a deterministic polynomial-time algorithm that computes
a pairwise disjoint family
\[
    \mathcal S \subseteq [m]
\]
satisfying
\[
    \sum_{i \in \calS} w_i \ge \frac{1}{k\Dmax} w(G).
\]
\end{lemma}

\begin{proof}
Let $\Delta(G)$ denote the maximum degree of the term-intersection
graph $G$.
Consider a vertex $i$ of the term-intersection graph. The local term
$h_i$ acts on at most $k$ qubits, and each such qubit participates in
at most $\Dmax$ local terms. Therefore, $h_i$ intersects at most
\[
    k(\Dmax-1)
\]
other local terms. Hence
\[
    \Delta(G) \le k(\Dmax-1)
\]
and therefore
\[
    \Delta(G)+1 \le k\Dmax.
\]

A greedy coloring algorithm consequently produces a proper coloring of $G$ using at most $k\Dmax$ colors.

The color classes partition the vertex set.
Therefore, at least one color class has total weight at least
\[
    \frac{1}{k\Dmax} w(G).
\]
Every color class is an independent set. Taking a color class of maximum total weight yields the desired family $\mathcal S$.
\end{proof}

\paragraph{$L_p$-based truncation construction}
\begin{lemma}
\label{lem:disjoint-family-Lp}
    For every $p > 1$, define
    \begin{align}
        D_p = \left( \frac{1}{n} \sum_{q = 1}^n d_q^p  \right)^{1/p}.
    \end{align}
    Then there exists a deterministic polynomial-time algorithm that constructs a pairwise disjoint family $\calS\subseteq[m]$ satisfying
    \begin{align}
        |\calS| \ge n \frac{p-1}{k^2 p} \left( \frac{\mu^p}{D_p^p p k} \right)^{\frac{1}{p-1}}.
    \end{align}
\end{lemma}
\begin{proof}
    For $\Theta \in \RR_{>0}$, define the set
    \begin{align}
        \Vhub(\Theta) = \{v \in [n] \mid d_v \ge \Theta\}.
    \end{align}
    Let $\mhub(\Theta)$ be the number of constraints $h_i$ acting on at least one qubit in $\Vhub(\Theta)$.
    We have
    \begin{align}
        \mhub(\Theta) \le \sum_{q \in \Vhub(\Theta)} d_q \le \sum_{q \in \Vhub(\Theta)} \frac{d_q^p}{\Theta^{p - 1}} \le \sum_{q = 1}^n  \frac{d_q^p}{\Theta^{p - 1}} = \frac{n D_p^p}{\Theta^{p - 1}}.
    \end{align}
    Now discard all constraints acting on at least one $q \in \Vhub(\Theta)$. By construction, in the resulting instance, every qubit $q$ has $d_q < \Theta$. In particular, in the resulting instance, every constraint intersects the support of at most $k\Theta$ other constraints, and thus the degree of the resulting term-intersection graph is at most $k \Theta$.
    Applying the maximum-degree construction (\cref{lem:weighted-disjoint-family-max-degree}) to the remaining term-intersection graph with unit vertex weights yields an independent set $\calS$ satisfying
    \begin{align}
        |\calS| \ge \frac{m - \mhub}{k\Theta} \ge \frac{1}{k \Theta} \left(\frac{n \mu}{k} - \frac{n D_p^p}{\Theta^{p-1}}  \right)= \frac{n}{k^2} \left( \frac{\mu}{\Theta}  - \frac{k D_p^p}{\Theta^p}\right).
    \end{align}
    Taking the derivative with respect to $\Theta$ and setting it to $0$ shows that
    \[
        \frac{n}{k^2}
        \left(
            \frac{\mu}{\Theta}
            -
            \frac{kD_p^p}{\Theta^p}
        \right)
    \]
    is maximized at
    \[
        \Theta_{\mathrm{opt}}
        =
        \left(
            \frac{pkD_p^p}{\mu}
        \right)^{\frac{1}{p-1}},
    \]
    which can be computed in polynomial time. Substituting $\Theta_{\text{opt}}$ we get
    \begin{align}
        |\calS| &\ge \frac{n}{k^2}\frac{\mu}{\Theta_{\text{opt}}}\left(1 - \frac{1}{p}\right) \\
        &\ge n \frac{p-1}{k^2 p} \left( \frac{\mu^p}{D_p^ppk}\right)^{\frac{1}{p-1}}.
    \end{align}
\end{proof}

\subsection{Main result}
\label{subsec:main result}

The notation $O^*(\cdot)$ suppresses factors polynomial in $n$ and $m$.
\begin{theorem}
\label{theorem:main result for upper bound}
    Let $H=\sum_{i=1}^{m}h_i$ be an instance of the frustration-free $k$-local Hamiltonian problem on $n$ qubits, where $k=O(1)$. For each qubit $q$, let $d_q = |\{i \mid q \in \operatorname{supp}(h_i)\}|$ be the number of constraints acting on $q$. Let $\mu = \frac{km}{n}$ upper bound the average number of constraints acting on a qubit and define the following qubit degree measures:
    \begin{align}
        \Dmax &= \max_q d_q \\
        D_p &= \left(\frac{1}{n} \sum_{q = 1}^n d_q^p \right)^{1/p},
    \end{align}
    where $p>1$. Consider the ratios
    \begin{align}
        \rhomax = \frac{\mu}{\Dmax}, \quad \rho_p = \frac{\mu}{D_p}.
    \end{align}
    Suppose $\rhomax \ge \varepsilon$ or that for some constant $p>1$, $\rho_p \ge \varepsilon$, where $\varepsilon > 0$ is a constant. Then there is a constant $\varepsilon'$ and a randomized classical algorithm solving the instance $H$ in time
    \begin{align}
        O^*\left( 2^{(1 - \varepsilon')n} \right).
    \end{align}
    Moreover, there then is a quantum algorithm solving the instance in time
    \begin{align}
        O^*\left( 2^{(1 - \varepsilon')\frac{n}{2}} \right).
    \end{align}
\end{theorem}

\begin{proof}
Suppose $\rho_p \ge \varepsilon$. Then, by
\cref{lem:disjoint-family-Lp}, we can construct an independent set
$\calS$ with
\begin{align}
    |\calS|
    &\ge
    n \frac{p-1}{k^2 p}
    \left(
        \frac{\rho_p^p}{pk}
    \right)^{\frac{1}{p-1}} \\
    &\ge
    n \frac{p-1}{k^2 p}
    \left(
        \frac{\varepsilon^p}{pk}
    \right)^{\frac{1}{p-1}}.
\end{align}
Since $w_s \ge \beta_k$ for every $s\in\calS$, we have
\begin{align}
    \sum_{s\in\calS} w_s
    &\ge \beta_k |\calS| \\
    &\ge
    \beta_k
    \frac{p-1}{k^2 p}
    \left(
        \frac{\varepsilon^p}{pk}
    \right)^{\frac{1}{p-1}} n.
\end{align}
Setting
\[
    \varepsilon'
    :=
    \beta_k
    \frac{p-1}{k^2 p}
    \left(
        \frac{\varepsilon^p}{pk}
    \right)^{\frac{1}{p-1}},
\]
we obtain
\[
    \sum_{s\in\calS} w_s \ge \varepsilon' n.
\]

If instead $\rhomax \ge \varepsilon$, we can use
\cref{lem:weighted-disjoint-family-max-degree} to construct a weighted
independent set $\calS$ with
\begin{align}
    \sum_{s\in\calS} w_s
    &\ge \frac{1}{k\Dmax} w(G) \\
    &\ge \frac{\beta_k}{k\Dmax} m \\
    &= \frac{\beta_k\rhomax}{k^2}n \\
    &\ge \frac{\beta_k\varepsilon}{k^2}n,
\end{align}

Setting $\varepsilon' = \frac{\beta_k \varepsilon}{k^2}$ we see that
\[
\sum_{s\in\calS} w_s \ge \varepsilon' n.
\]

In either case, applying \cref{thm:compression} gives an effective Hamiltonian
\[
    H_{\mathrm{eff}} = V^\dagger H V
\]
acting on a Hilbert space of dimension
\[
    M=2^{\,n-\sum_{s\in\mathcal S}w_s} \le 2^{(1 - \varepsilon')n}.
\]

The effective Hilbert space has the tensor-product form
\[
    \left(\bigotimes_{s\in\mathcal S}\mathbb C^{g_s}\right) \otimes (\mathbb C^2)^{\otimes\left(n-\sum_{s\in\mathcal S}k_s\right)},
\]
where $g_i:=\dim\ker(h_i)$.
Each tensor factor has constant dimension because all $h_i$ are $k$-local and thus
\[
    g_s \le 2^k =O(1).
\]
Moreover, each projected local term
\[
    V^\dagger h_\eta V
\]
acts non-trivially on at most $k=O(1)$ tensor factors. Thus $H_{\mathrm{eff}}$ satisfies the sparsity assumptions of \cref{fact:exponential-time algo}.
Therefore, the smallest eigenvalue of $H_{\mathrm{eff}}$ can be estimated to inverse-polynomial additive error classically in time
\[
    O^*(M)
\]
and quantumly in time
\[
    O^*(\sqrt{M}),
\]
with high probability.

By \cref{thm:compression}, if
\[
    \lambda(H)=0,
\]
then
\[
    \lambda(H_{\mathrm{eff}})=0.
\]
Moreover,
\[
    \lambda(H_{\mathrm{eff}}) \ge \lambda(H).
\]
Therefore, in the NO case, the ground-state energy of $H_{\mathrm{eff}}$ remains separated from zero by the original inverse-polynomial promise gap. Hence an inverse-polynomial additive approximation to $\lambda(H_{\mathrm{eff}})$ solves the original instance.

Finally, substituting the upper bound on $M$ gives the claimed classical and quantum running times.
\end{proof}

\subsection{Interesting parameter regimes}
We now discuss implications of our main result for concrete examples.
\paragraph{Almost regular instances}
One regime where \cref{theorem:main result for upper bound} gives an improvement over \cref{fact:exponential-time algo} is when the interaction graph is almost regular. Specifically, we have:
\begin{corollary}
    Consider quantum $k$-SAT. Let $n$ be the number of qubits and $m$ the number of constraints. Suppose we are promised that there exists some constant $\alpha\ge 1$ such that every qubit is acted upon by at most $\alpha \mu$ constraints, where $\mu = \frac{km}{n}$ upper bounds the average number of constraints per qubit.
    With this additional promise, quantum $k$-SAT can be solved in randomized classical time
    \begin{align}
        O^*\left( 2^{(1 - \varepsilon')n} \right)
    \end{align}
    and in quantum time
    \begin{align}
        O^*\left( 2^{(1 - \varepsilon')\frac{n}{2}} \right)
    \end{align}
    for some constant $\varepsilon' > 0$.
\end{corollary}
\begin{proof}
    By the promise, we have
    \begin{align}
        D_p = \left(\frac{1}{n} \sum_{q = 1}^n d_q^p\right)^{1/p} \le \left(\frac{1}{n} \sum_{q = 1}^n (\alpha \mu)^p\right)^{1/p} \le \alpha \mu.
    \end{align}
    Therefore, $\rho_p \geq \frac{1}{\alpha}$. Applying \cref{theorem:main result for upper bound} with $\varepsilon = \frac{1}{\alpha}$ now gives the claimed result. 
\end{proof}
Note that, in particular, this almost regular case includes the case where each qubit is acted upon at most $O(1)$ times.

\paragraph{Small number of high degree qubits}
Our algorithm also works in the more general case where the qubits can be divided into subsets, where for each, the degree of the qubits is at most constant, or the size of subset multiplied by the maximal degree of the qubits in the set is $O(n^{1 - \varepsilon})$ for some $\varepsilon > 0$. 
\begin{corollary}
    \label{cor:smallnumberhighdegree}
    Consider quantum $k$-SAT. Let $n$ be the number of qubits, $m$ the number of constraints and set $\mu = \frac{km}{n}$. Suppose we are promised that there exists a partition of $[n]$ into $r = O(1)$ disjoint subsets $[n] = L \cup \bigcup_{i = 1}^r S_i$, where for every $i$ there is a size bound $s_i \ge 0$ and a degree bound $\delta_i \ge 0$ such that $|S_i| \le O(n^{s_i})$ and
    \begin{align}
        q \in L &\implies d_q \le O(\mu), \\
        q \in S_i &\implies d_q \le O(n^{\delta_i} \mu).
    \end{align}
    If for all $i$, $s_i + \delta_i < 1$, then this instance of quantum $k$-SAT can be solved in randomized classical time
    \begin{align}
        O^*\left( 2^{(1 - \varepsilon')n} \right)
    \end{align}
    and in quantum time
    \begin{align}
        O^*\left( 2^{(1 - \varepsilon')\frac{n}{2}} \right)
    \end{align}
    for some constant $\varepsilon' > 0$.
\end{corollary}
\begin{proof}
    We have
    \begin{align}
        D_p = \left[\frac{1}{n} \left(\sum_{q \in L} d_q^p + \sum_{i = 1}^r \sum_{q \in S_i} d_q^p \right) \right]^{1/p}.
    \end{align}
    Let $\alpha \ge 1$ be such that $d_q \le \alpha \mu$ when $q \in L$, $d_q \le \alpha n^{\delta_i} \mu$ when $q \in S_i$ and $|S_i| \le \alpha n^{s_i}$. The promise guarantees the existence of such an $\alpha$.
    Substituting, we get
    \begin{align}
        D_p &\le \left(\frac{|L|}{n} \alpha^p \mu^p + \sum_{i = 1}^r \frac{\alpha n^{s_i}}{n} \alpha^p n^{\delta_i p} \mu^p \right)^{1/p} \\
        &\le \alpha \mu + \sum_{i = 1}^r \alpha^{\frac{p+1}{p}}n^{\frac{s_i - 1}{p} + \delta_i} \mu.  
    \end{align}
    Set $p = 1 + x$, then
    \begin{align}
        \frac{s_i - 1}{ p} + \delta_i = \frac{s_i -1 + p\delta_i}{p} = \frac{\delta_i + s_i - 1 + \delta_i x}{p}.
    \end{align}
    By setting $0< x \le \min_i 1 - \delta_i - s_i$ (possible by assumption) we get 
    \begin{align}
        \frac{s_i - 1}{ p} + \delta_i \le 0
    \end{align}
    for all $i$, which implies
\[
    D_p
    \le
    \alpha\mu
    +
    r\alpha^2\mu
    \le
    (r+1)\alpha^2\mu.
\]
Hence,
\[
    \rho_p
    =
    \frac{\mu}{D_p}
    \ge
    \frac{1}{(r+1)\alpha^2}.
\]
Applying \cref{theorem:main result for upper bound} with $\varepsilon=\frac{1}{(r+1)\alpha^2}$ now completes the proof.
\end{proof}
\section{(Q)SETH hardness of quantum 5-SAT}
In this section, we show that in contrast to \cref{theorem:main result for upper bound}, assuming the Strong Exponential Time Hypothesis, general $5$-local quantum SAT instances on $n$ qubits cannot be solved in time $O(2^{(1 - \varepsilon) n})$ for any $\varepsilon > 0$. This result is largely based on work by Chia, Hasegawa, Le Gall and Shen \cite{ChiaHLGY26} who prove a similar statement for the 3-local Hamiltonian problem.

The central idea behind their result is the \emph{size-preserving circuit-to-Hamiltonian construction}. This allows embedding a circuit verifying $k$-SAT on $n$ variables (obtained from \cref{lem:U_compute_Phi} or \cref{thm:two-unary-sat-verifier}) into a Hamiltonian on $n + o(n)$ qubits as opposed to the $O(n + T)$ qubits achieved by standard circuit-to-Hamiltonian constructions ($T$ is the runtime of the circuit). An algorithm solving local Hamiltonian on $n$ qubits in time $2^{(1 - \varepsilon)n}$ for $\varepsilon > 0$ thus gives rise to an algorithm solving $k$-SAT on $n$ variables in time $2^{(1 - \varepsilon')n}$ for some $\varepsilon' > 0$, which contradicts SETH.

The main technical difficulty is that the size-preserving circuit-to-Hamiltonian construction from \cite{ChiaHLGY26} does not produce valid QSAT instances. In this section we describe how to modify their construction to produce QSAT instances. This modification does require a small locality increase: we obtain SETH hardness for quantum $5$-SAT as opposed to the SETH hardness for $3$-LH obtained by \cite{ChiaHLGY26}.

For their construction, Chia et al. use a two-part clock, representing time in ``$\sqrt{T}$-ary''. That is, they represent the time $t = (k-1)\sqrt{T} + \ell$ as $\ket{c_k}\ket{c_\ell}$ where both parts live on $O(\sqrt{T})$ qubits. For the first clock state $\ket{c_k}$, that is, for the ``leading digit'', they use a ``pulse clock'': a clock with a single 1 whose position represents the time step. This allows them to ``read out'' the time of the clock from a single qubit. For the second clock state they use a unary clock. Here, the operator $\ket{1}\bra{0}_{\ell+1}$ allows incrementing the time with a term of locality $1$.

Two immediate obstacles present themselves when converting this approach to QSAT. Firstly, the Hamiltonian to enforce the pulse clock given in \cite{ChiaHLGY26} is not frustration-free. Secondly, the operator $\ket{1}\bra{0}_{\ell+1}$ for incrementing the unary clock is not a projector and hence not a legal QSAT constraint. 

To resolve these issues we use two different clocks. For the ``leading digit'' $k$ we use a pulse clock, but couple this to a unary clock so the legal clock states can be enforced by local projectors. For the ``less significant digit'' $\ell$ we use a clock due to Bravyi \cite{bravyi06} that allows for 2-local incrementing using projectors. We now describe the clocks in more detail.

\subsection{The $K$ clock}
The main property that we want the $K$ clock to have is that it can be read out 1-locally. That is, we want to construct the $K$ clock such that its state can be uniquely determined by a 1-local projector. For this we want to use the pulse clock $00\dots010\dots 0$ just like \cite{ChiaHLGY26}. However, enforcing that there is at least one $1$ does not seem possible with local projectors (Chia et al. use frustrated terms). To get around this we couple to a unary clock. That is, we consider two registers: $P,U$, both of $S_K$ qubits, where $S_K$ is the maximum time in the clock. The first contains the pulse clock part and the second the unary part. The intended clock states will be:
\begin{align}
    \ket{\CK_1} &= \ket{10\dots 0}_P \ket{10\dots}_U \\
    \ket{\CK_2} &= \ket{010\dots 0}_P \ket{110\dots}_U \\
    \ket{\CK_3} &= \ket{0010\dots 0}_P \ket{1110\dots}_U \\
    &\vdots \notag \\
    \ket{\CK_s} &= \ket{0^{s-1}10^{S_K-s}}_P \ket{1^s0^{S_K-s}}_U \\
    &\vdots \notag \\
    \ket{\CK_{S_K}} &= \ket{0^{S_K-1}1}_P \ket{1^{S_K}}_U.
\end{align}
This clock can indeed be enforced with $3$-local projectors.
To enforce the unary structure of the $U$ register, we use the projector
\begin{align}
    \label{eq:Hunary}
    \Hunary = \sum_{i = 1}^{S_K-1} \ketbrab{01}_{U_i, U_{i+1}}.
\end{align}
We will also demand that the first qubit of $U$ is a $1$ through the initialization projector:
\begin{align}
    \label{eq:Hinit}
    \Hinit = \ketbrab{0}_{U_1}.
\end{align}
Next, we enforce that there is at most a single $1$ in the $P$ register:
\begin{align}
    \label{eq:Hlesstwo}
    \Hlesstwo = \sum_{1\le i < j \le S_K} \ketbrab{11}_{P_i P_j}.
\end{align}
Note that there could be no 1s whatsoever, which we want to avoid. This is why we couple with a unary clock through the following coupling projector:
\begin{align}
    \label{eq:Hsync}
     \begin{split}
        \Hsync =&\sum_{i = 1}^{S_K-1} \left(\ketbrab{100} + \ketbrab{010} + \ketbrab{111} \right)_{P_i U_i U_{i+1}} \\
        &+ \ketbrab{01}_{P_{S_K}U_{S_K}} + \ketbrab{10}_{P_{S_K}U_{S_K}} 
    \end{split}
\end{align}
\begin{claim}
    The zero-energy space of the full pulse + unary clock Hamiltonian $\HPU = \Hunary + \Hinit + \Hlesstwo + \Hsync$ is spanned by the states $\ket{\CK_i}$.
\end{claim}
\begin{proof}[Proof of claim]
    Note that $\Hunary$ ensures that there is no $0$ followed by a $1$ in the $U$ register; in other words, it enforces the unary structure. $\Hinit$ then excludes states where the $U$ register is all $0$s. On the $P$ register, states with $2$ or more $1$'s are excluded by $\Hlesstwo$.
    
    Finally, consider $\Hsync$. Note that this prohibits qubits $P_i, U_i, U_{i+1}$ to be in the states $\ket{100}, \ket{010}$ and $\ket{111}$. The states $\ket{101}$ and $\ket{001}$ are already excluded by $\Hunary$. It follows that the only legal assignments are $\ket{000}$, $\ket{011}$ and $\ket{110}$. In other words, it is required that there is a $1$ in the $i$-th qubit of the $P$ register, where $i$ is the end of the string of $1$s in the unary part. Note that the additional terms on $P_{S_K}, U_{S_K}$ similarly enforce this for the final time step. In particular, (together with $\Hinit$) $\Hsync$ enforces that there is always a $1$ in the $P$ register.  
\end{proof}

The time of the $K$ clock can now be read out using a $1$-local projector. That is, for $\PiPU$ the projector on the zero-energy space of $\HPU$ we have
\begin{align}
    \label{eq:clockreadK}
    \PiPU \ketbrab{1}_{P_k} \PiPU = \ketbrab{\CK_k}.
\end{align}
Incrementing the $K$ clock is made more difficult by our coupling to the unary register, but can still be done with the following $3$-local projector. 
\begin{align}
    \PiKinc_{k,k+1} = \frac{1}{2}\big(\ket{100} - \ket{011} \big)\big(\bra{100} - \bra{011}\big)_{P_k P_{k+1} U_{k+1}}.
\end{align}
We have
\begin{align}
\label{eq:clockincK}
\begin{split}
    \PiPU \PiKinc_{k,k+1} \PiPU =& \frac{1}{2} \big(\ketbrab{\CK_k} + \ketbrab{\CK_{k+1}} \\ &- \ket{\CK_k}\bra{\CK_{k+1}} - \ket{\CK_{k+1}}\bra{\CK_{k}}\big).
\end{split}
\end{align}

\subsection{The $L$ clock}
For the $L$ clock the main factor of importance is the locality of the incrementing operation. 
We will use a clock due to Bravyi \cite{bravyi06} that allows $2$-local incrementing.
The clock register is defined as
\begin{align}
    \mathcal H_{\mathrm{clock}} = \left(\mathbb C^4\right)^{\otimes S_L}.
\end{align}
Here $S_L$ is the number of computational gate applications we want to encode. There will be more time-steps ($2S_L$ in fact), but not all transitions will have enough ``locality spare'' to apply a computational gate.

Each clock particle has the orthonormal basis
\begin{align*}
    \{\ket{u},\ket{a_1},\ket{a_2},\ket{d}\}
\end{align*}
where $u$, $a_1$, $a_2$, and $d$ stand for
\emph{unborn}, \emph{active phase 1}, \emph{active phase 2}, and \emph{dead}, respectively.
Since each clock particle is four-dimensional, it can be represented by two qubits. We use the following encoding
\begin{align}
    \ket{u} &= \ket{00} \\
    \ket{a_1} &= \ket{10} \\
    \ket{a_2} &= \ket{11} \\
    \ket{d} &= \ket{01}.
\end{align}
Valid clock states are supposed to be of the form
\begin{align}
    \ket{\CL_1} &= \ket{a_1}\ket{u^{S_L - 1}} \\
    \ket{\DL_1} &= \ket{a_2}\ket{u^{S_L - 1}} \\
    \ket{\CL_2} &= \ket{d}\ket{a_1}\ket{u^{S_L - 2}} \\
    &\;\vdots \\
    \ket{\CL_\ell} &= \ket{d}^{\otimes \ell - 1} \ket{a_1} \ket{u}^{\otimes S_L - \ell} \\
    \ket{\DL_\ell} &= \ket{d}^{\otimes \ell - 1} \ket{a_2} \ket{u}^{\otimes S_L - \ell} 
\end{align}

To enforce the clock states are in legal time steps, Bravyi uses the following sum of projectors:
\begin{align}
\label{eq:HBravyi}
\begin{split}
    \Hbrav_L =& \ketbrab{u}_1 + \ketbrab{d}_{S_L} \\
    &\;+ \sum_{1 \le i < j \le S_L} \big(\ketbrab{a_1} + \ketbrab{a_2}\big)_{2i - 1, 2i} \otimes \big(\ketbrab{a_1} + \ketbrab{a_2}\big)_{2j - 1, 2j} \\
    &\;+ \sum_{1 \le i < j \le S_L} \big(\ketbrab{a_1} + \ketbrab{a_2} + \ketbrab{u}\big)_{2i - 1, 2i} \otimes \ketbrab{d}_{2j-1, 2j} \\
    &\;+ \sum_{1 \le i < j \le S_L} \ketbrab{u}_{2i - 1, 2i} \otimes \big(\ketbrab{a_1} + \ketbrab{a_2} + \ketbrab{d}\big)_{2j - 1, 2j} \\
    &\;+ \sum_{1 \le i \le S_L - 1} \ketbrab{d}_{2_i - 1, 2i} \otimes \ketbrab{u}_{2i+1, 2i+2}.
\end{split}
\end{align}
The terms ensure, in order, that
\begin{itemize}
    \item The first qubit pair is never $\ket{u}$ (the state cannot be $\ket{u^{S_L}}$).
    \item The last qubit pair is never $\ket{d}$ (the state cannot be $\ket{d^{S_L}}$).
    \item There is at most one active state.
    \item Only $\ket{d}$ can appear before $\ket{d}$ (the $d$'s form a prefix).
    \item Only $\ket{u}$ can appear after $\ket{u}$ (the $u$'s form a suffix).
    \item $\ket{u}$ cannot immediately follow $\ket{d}$ (there has to be an active state in between).
\end{itemize}

Bravyi uses two types of projectors to increment the clock: a $4$-local ``shift'' transition from $\ket{\DL_\ell}$ to $\ket{\CL_{\ell + 1}}$ and a $2$-local transition from $\ket{\CL_\ell}$ to $\ket{\DL_\ell}$ during which a computational gate can be applied. The transition operators are
\begin{align}
    \PiLshift_{\ell, \ell+1} &= \frac{1}{2}\big(\ket{a_2 u} - \ket{d a_1}\big)\big(\bra{a_2 u} - \bra{d a_1}\big)_{2\ell - 1, 2\ell, 2\ell + 1, 2\ell +2}\\
    &= \frac{1}{2}\big(\ketbrab{a_2 u} + \ketbrab{d a_1} - \ket{a_2 u}\bra{d a_1} - \ket{d a_1}\bra{a_2 u} \big)_{2\ell - 1, 2\ell, 2\ell + 1, 2\ell +2}   \\
    \PiLgate_{\ell}(U) &= \frac{1}{2}\left( I \otimes \ketbrab{a_1} + I \otimes \ketbrab{a_2} - U^\dagger \otimes \ket{a_1}\bra{a_2} - U \otimes \ket{a_2} \bra{a_1}\right), 
\end{align}
where the first terms of $\PiLgate$ act on the computational space, and the second on the $2\ell-1$ and $2\ell$-th qubits of the $L$ register. Note that $\PiLshift$ and $\PiLgate$ are indeed projectors and that their locality is $4$ when $U$ is $2$-local. Restricted to the ground state of $\Hbrav$ we have (see \cite[Eqns. 22\&23]{bravyi06}),
\begin{align}
    \Pibrav \PiLshift_{\ell, \ell + 1} \Pibrav = \frac{1}{2}\big(\ketbrab{\DL_\ell} + \ketbrab{\CL_{\ell + 1}} - \ket{\DL_{\ell}}\bra{\CL_{\ell + 1}} - \ket{\CL_{\ell + 1}}\bra{\DL_\ell}\big) \label{eq:clockshiftL}
\end{align}
and 
\begin{align}
\begin{split}
    \big(I_C \otimes \Pibrav\big) \PiLgate_{\ell}(U) \big(I_C \otimes \Pibrav) =& \frac{1}{2} \big(I \otimes \ketbrab{\CL_\ell} + I \otimes \ketbrab{\DL_{\ell}}\\ &\;- U \otimes \ket{\DL_{\ell}}\bra{\CL_{\ell}} - U^\dagger \otimes \ket{\CL_{\ell}}\bra{\DL_\ell}\big). \label{eq:clockgateL}
\end{split}
\end{align}
Here $\Pibrav$ is the projector on the ground space of $\Hbrav$.

\subsection{Compressed clock construction}
We now combine both clocks into a single compressed clock with $T = S_KS_L$ time steps. The idea is to keep incrementing the $L$ clock until it is full, and then increment the $K$ clock, reset the $L$ clock and repeat. As resetting the $L$ clock turns out to be difficult, we use a ``snaking'' technique instead. That is, we alternate incrementing and decrementing the $L$ clock. 

The total Hamiltonian will be
\begin{align}
\label{eq:Hcompressedclock}
    \Hcomp = \HPU_K \otimes I_L + I_K \otimes \Hbrav_L.
\end{align}
Clearly, the ground state is spanned by states of the form $\ket{\CK_k}\ket{\CL_\ell}$ and $\ket{\CK_k}\ket{\DL_\ell}$ with $k \le S_K$ and $\ell \le S_L$. We will define the compressed time steps as
\begin{align}
    \ket{\Ccomp_1} &= \ket{\CK_1}\ket{\CL_1}  \\
    \ket{\Dcomp_1} &= \ket{\CK_1}\ket{\DL_1} \\
    \ket{\Ccomp_2} &= \ket{\CK_1}\ket{\CL_2} \\
    &\;\vdots\nonumber \\
    \ket{\Ccomp_{S_L}} &= \ket{\CK_1}\ket{\CL_{S_L}} \\
    \ket{\Dcomp_{S_L}} &= \ket{\CK_1}\ket{\DL_{S_L}} \\
    \ket{\Ccomp_{S_L + 1}} &= \ket{\CK_2}\ket{\DL_{S_L}} \\
    \ket{\Dcomp_{S_L + 1}} &= \ket{\CK_2}\ket{\CL_{S_L}} \\
    \ket{\Ccomp_{S_L + 2}} &= \ket{\CK_2}\ket{\DL_{S_L - 1}} \\
    &\;\vdots \nonumber \\
    \ket{\Dcomp_{T}} &= \ket{\CK_{S_K}}\ket{\DL_{S_L}}. \label{eq:finaltimestep}
\end{align}
That is, writing $t = (k - 1)S_L + \ell$, we have
\begin{align}
\label{eq:Ccomp}
    \ket{\Ccomp_t} = \begin{cases}
        \ket{\CK_k}\ket{\CL_\ell} \quad&\text{when } k \text{ is odd}\\
        \ket{\CK_k}\ket{\DL_{S_L - \ell + 1}} \quad&\text{when } k \text{ is even},
    \end{cases}
\end{align}
and
\begin{align}
\label{eq:Dcomp}
    \ket{\Dcomp_t} = \begin{cases}
        \ket{\CK_k}\ket{\DL_\ell} \quad&\text{when } k \text{ is odd}\\
        \ket{\CK_k}\ket{\CL_{S_L - \ell + 1}} \quad&\text{when } k \text{ is even}.
    \end{cases}
\end{align}
Note that, perhaps confusingly, the definition of $\ket{\Ccomp_t}$ involves $\ket{\DL_\ell}$ and vice-versa when $k$ is even. In \cref{eq:finaltimestep} we assumed that $S_K$ is odd (which is without loss of generality).
 
\paragraph{Incrementing the compressed clock}
We will again define two separate incrementing projectors. One ``shift'' operator incrementing $\ket{\Dcomp_t}$ to $\ket{\Ccomp_{t+1}}$ and a ``gate'' operator incrementing $\ket{\Ccomp_t}$ to $\ket{\Dcomp_t}$ while applying the unitary $U$ to the computational register.

There are several different cases that need to be distinguished. We begin with the most straightforward case, where $L$ is being incremented and has not reached the end. This corresponds with odd $k$ and $\ell < S_L$. In this case
\begin{align}
    \Pishift_{t,t+1} &= \ketbrab{1}_{P_k} \otimes \PiLshift_{\ell, \ell + 1} \\
    \Pigate_t(U) &= \ketbrab{1}_{P_k} \otimes \Pigate_\ell(U)
\end{align}
Denoting by $\Picomp$ the projector on the ground space of $\Hcomp$, we have
\begin{align}
\label{eq:shiftodd}
\begin{split}
    \Picomp \Pishift_{t,t+1} \Picomp =& \frac{1}{2}\big(\ketbrab{\Dcomp_t} + \ketbrab{\Ccomp_{t+1}} \\ &\;- \ket{\Dcomp_t}\bra{\Ccomp_{t+1}} - \ket{\Ccomp_{t+1}}\bra{\Dcomp_{t}}\big), 
\end{split}
\end{align}
and
\begin{align}
\label{eq:gateodd}
    \begin{split}
    \big(\Picomp \otimes I_{C}\big) \Pigate_t(U) \big(\Picomp \otimes I_{C}\big) &= \frac{1}{2} \big( \ketbrab{\Ccomp_t} \otimes I + \ketbrab{\Dcomp_t} \otimes I \\
        &\;- \ket{\Dcomp_t}\bra{\Ccomp_t} \otimes U - \ket{\Ccomp_t}\bra{\Dcomp_t}\otimes U^\dagger\big).
    \end{split}
\end{align}
as can be seen by combining \cref{eq:clockreadK,eq:clockshiftL,eq:clockgateL}.

When $L$ is being decremented but has not reached the end, that is, when $k$ is even and $\ell < S_L$, the projectors are similar
\begin{align}
    \Pishift_{t,t+1} &= \ketbrab{1}_{P_k} \otimes \PiLshift_{S_L - \ell, S_L - \ell + 1} \\
    \Pigate_t(U) &= \ketbrab{1}_{P_k} \otimes \Pigate_{S_L - \ell + 1}(U^\dagger).
\end{align}
Combining \cref{eq:clockreadK,eq:clockshiftL,eq:clockgateL} with the $k$ is even case of \cref{eq:Ccomp,eq:Dcomp} now yields
\begin{align}
\begin{split}
    \Picomp \Pishift_{t,t+1} \Picomp &= \frac{1}{2} \ketbrab{\CK_k} \otimes \Big(\ketbrab{\DL_{S_L - \ell}} + \ketbrab{\CL_{S_L - \ell + 1}} \\&\qquad - \ket{\DL_{S_L - \ell}}\bra{\CL_{S_L - \ell + 1}} - \ket{\CL_{S_L - \ell + 1}}\bra{\DL_{S_L -\ell}} \Big)
\end{split}
     \\ &= \frac{1}{2}\big( \ketbrab{\Dcomp_t} + \ketbrab{\Ccomp_{t+1}} - \ket{\Dcomp_t}\bra{\Ccomp_{t+1}} - \ket{\Ccomp_{t+1}}\bra{\Dcomp_{t}}\big), \label{eq:shifteven}
\end{align}
and
\begin{align}
\label{eq:gateeven}
    \begin{split}
    \big(\Picomp \otimes I_{C}\big) \Pigate_t(U) \big(\Picomp \otimes I_{C}\big) =& \frac{1}{2}\big(\ketbrab{\Ccomp_t} \otimes I + \ketbrab{\Dcomp_t} \otimes I \\
        &\;- \ket{\Dcomp_t}\bra{\Ccomp_t} \otimes U - \ket{\Ccomp_t}\bra{\Dcomp_t}\otimes U^\dagger\big).
    \end{split}
\end{align}

Finally, we need projectors for incrementing $K$ when $L$ has reached the end. Again, this will depend on whether $k$ is odd or even.
\begin{align}
    \Pishift_{t,t+1} = \begin{cases}
        \PiKinc_{k,k+1} \otimes \ketbrab{a_2}_{2S_L - 1, 2S_L} \quad& \text{when } k \text{ is odd} \\
        \PiKinc_{k,k+1} \otimes \ketbrab{a_1}_{1, 2} \quad& \text{when } k \text{ is even.}
    \end{cases}
\end{align}
Note that we indeed have
\begin{align}
\label{eq:shiftend}
\begin{split}
    \Picomp \Pishift_{t,t+1} \Picomp =& \frac{1}{2}\big(\ketbrab{\Dcomp_{t}} + \ketbrab{\Ccomp_{t+1}} \\&\;- \ket{\Dcomp_t}\bra{\Ccomp_{t+1}} - \ket{\Ccomp_{t+1}}\bra{\Dcomp_t}\big).
\end{split}
\end{align}

\paragraph{Finalizing the compressed clock construction}
Having constructed our $5$-local compressed clock and the corresponding propagation operators $\Pishift, \Pigate$, it only remains to define an in- and output term before we can apply the standard circuit-to-Hamiltonian construction to obtain our result. These are defined straightforwardly:
\begin{align}
    \Hin &= \sum_{i} \ketbrab{1}_{P_{1}} \otimes \ketbrab{a_1}_{1,2} \otimes \ketbrab{1}_{\text{anc}_i} \\
    \Hout &= \ketbrab{1}_{P_{S_K}} \otimes \ketbrab{a_2}_{2S_L - 1, 2S_L} \otimes \ketbrab{0}_{C_{\text{out}}},
\end{align}
where the $\text{anc}_i$ denote the ancilla qubits of the computational register and $C_{\text{out}}$ denotes the designated output qubit.

\begin{theorem}
\label{thm:compressedcirctoQSAT}
    Let $U = U_T \dots U_1$ be a quantum verification circuit consisting of $T$ 2-local gates applied in succession to a workspace of $n$ qubits. Let $S_k, S_L$ be such that $S_KS_L = T$. Then, there exists a sum of $5$-local projectors $H$ acting on $n + O(S_K) + O(S_L)$ qubits, such that 
    \begin{itemize}
        \item If there is some $\ket{\psi}$ that is accepted by $U$ with probability $1$, then $H$ has a zero-energy ground state.
        \item If all states $\ket{\psi}$ are accepted by $U$ with probability at most $1- \varepsilon$, then $\lambda(H) \ge \Omega\left(\frac{{\varepsilon}}{T^3}\right)$.
    \end{itemize}
    In particular, setting $S_K = S_L = \sqrt{T}$ the Hamiltonian acts on $n + O(\sqrt{T})$ qubits. Note that using padding we can assume $\sqrt{T}$ to be an integer without loss of generality. 
\end{theorem}
\begin{proof}[Proof sketch]
    Consider registers $K, L, C$, where $K$ and $L$ are clock registers consisting of $2S_K$ and $2S_L$ qubits, respectively, and $C$ is the computational register consisting of $n$ qubits. We split the $K$ register into two parts $P,U$ of $S_K$ qubits.

    Let the sums of projectors $\Hcomp, \Hin, \Hout, \Pigate(U)$ and $\Pishift$ acting on the registers $K,L,C$ be as defined earlier in this section and consider the Hamiltonian
    \begin{align}
        H = \Hcomp_{KL} \otimes I_C + \Hin + \Hout + \Pigate_T(U_{T}) + \sum_{t = 1}^{T - 1} \left(\Pigate_t(U_{t}) + \Pishift_{t, t+1}\right).
    \end{align}
    
    We have shown that the $\Pigate$ and $\Pishift$ projectors act appropriately on the legal clock space in \cref{eq:shiftodd,eq:gateodd,eq:shifteven,eq:gateeven,eq:shiftend}. A standard analysis of circuit-to-Hamiltonian constructions (see \cite{Kitaev+02,bravyi06}) thus shows that if $\ket{\psi}$ is accepted by $U$ with probability $1$ then 
    \begin{align}
    \begin{split}
        \ket{\psi_{hist}} \propto& \ket{\Ccomp_1}\ket{\psi}_{C_{in}}\ket{\vec{0}}_{C_{anc}} + \ket{\Dcomp_T} \otimes U_T \dots U_1\ket{\psi}_{C_{in}}\ket{\vec{0}}_{C_{anc}} \\ &\;+  \sum_{t = 1}^{T - 1} \big(\ket{\Dcomp_t} + \ket{\Ccomp_{t + 1}}\big) \otimes U_t \dots U_1\ket{\psi}_{C_{in}}\ket{\vec{0}}_{C_{anc}}
    \end{split}
    \end{align}
    is a zero-energy state of $H$. Conversely, if all states are rejected by $U$ with probability at least $\varepsilon$ then $\lambda(H) \ge \Omega\left(\frac{{\varepsilon}}{T^3}\right)$.
\end{proof}

\paragraph{Counting qubit degrees}
To prove the optimality of our upper bound algorithm in terms of regularity, it will be useful to count how often $H$ acts on each of the qubits.

From \cref{eq:Hunary,eq:Hinit,eq:Hlesstwo,eq:Hsync} it is clear that $\HPU$ acts on each of the $P$ and $U$ qubits at most $O(S_K)$ times. Similarly, inspection of \cref{eq:HBravyi} shows that $\Hbrav_L$ acts on every $L$ qubit at most $O(S_L)$ times.

For the incrementing operators, note that the qubit encoding the value $k$ (so $P_k$) is acted upon $O(S_L)$ times: $O(1)$ times for each of the $S_L$ values of $\ell$. Similarly, the qubits in $L$ get acted upon $O(1)$ times for every value of $k$, so $O(S_K)$ times in total.

$\Hin$ acts upon the $P_1$ and $L_1, L_2$ qubits once for every ancilla qubit. Finally, $\Hout$ acts on all qubits at most once.

The degree of the computational qubits depends fully on the encoded circuit.

\subsection{(Q)SETH hardness of quantum $5$-SAT}
We are now ready to prove our second main result
\begin{theorem}[SETH hardness of quantum SAT]
\label{thm:sethhard}
    Suppose there is a classical algorithm solving quantum $5$-SAT on $n$ qubits in time $\myO{2^{(1-\varepsilon)n}}$ for some $\varepsilon >0$, then there exists some $\varepsilon' > 0$ and for any $k$ a classical algorithm solving $k$-SAT with $n$ variables and $O(n)$ clauses in time $\myO{2^{(1-\varepsilon')n}}$, contradicting the Strong Exponential Time Hypothesis (\cref{conjecture:SETH}).

    Similarly, a quantum algorithm for quantum $5$-SAT using time $\myO{2^{(1-\varepsilon)n/2}}$ gives a quantum algorithm for $k$-SAT using time $\myO{2^{(1-\varepsilon')n/2}}$ for some $\varepsilon' > 0$, thus contradicting the Quantum Strong Exponential Time Hypothesis (\cref{conjecture:QSETH}).
\end{theorem}
\begin{proof}
    Start by assuming the existence of an algorithm solving quantum $5$-SAT on $n$ qubits in time $\myO{2^{(1-\varepsilon)n}}$ for some $\varepsilon >0$. Let $k$ be arbitrary and consider the following algorithm solving $k$-SAT:
    \begin{enumerate}
        \item Let $\Phi$ be the input formula and let $n$ be the number of variables of $\Phi$. Use \cref{lem:U_compute_Phi} to compute the quantum circuit $U_\Phi$ that, on input $x$ outputs $\Phi(x)$ with probability $1$. The circuit acts on $n$ input qubits and $O(\log n)$ ancilla qubits, and consists of $T = O(n)$ elementary gates.
        \item Use the compressed circuit-to-QSAT construction from \cref{thm:compressedcirctoQSAT} with $S_K = S_L = \sqrt{T}$ on $U$ to construct a Hamiltonian $H$ such that $\lambda(H) = 0$ if $\Phi$ has a satisfying assignment and $\lambda(H) \ge \Omega\left(\frac{1}{T^3}\right)$ if $\Phi$ is unsatisfiable. The Hamiltonian $H$ will be a sum of $5$-local projectors acting on $N = n + O(\log n) + O(\sqrt{T}) = n + o(n) = n(1 + o(1))$ qubits.
        \item Use the hypothetical algorithm for quantum $5$-SAT to decide whether $\lambda(H) = 0$ or $\lambda(H) \ge \frac{1}{\poly(n)}$, and thus decide whether $\Phi$ is satisfiable or not. By assumption, this can be done in time $\myO{2^{(1-\varepsilon)N}}$. In particular, for any $\varepsilon > \varepsilon' > 0$ the runtime is $\myO{2^{(1-\varepsilon')n}}$ for sufficiently large $n$.
    \end{enumerate}
    Note that the $\myO{2^{(1-\varepsilon')n}}$ runtime holds for any $k$. We have thus contradicted the Strong Exponential Time Hypothesis (\cref{conjecture:SETH}) as desired.

    Replacing the classical $\myO{2^{(1-\varepsilon)n}}$ time algorithm for QSAT with a $\myO{2^{(1-\varepsilon)n/2}}$ quantum algorithm yields a $\myO{2^{(1-\varepsilon')n/2}}$ time quantum algorithm for $k$-SAT, which contradicts QSETH (\cref{conjecture:QSETH}). 
\end{proof}

\subsection{(Q)SETH hardness with somewhat balanced degree}
In order to contrast \cref{cor:smallnumberhighdegree}, we would like to prove (Q)SETH hardness for as regular instances as possible. The construction from \cref{thm:sethhard} is already pretty good in this regard: the $O(\sqrt{n})$ clock qubits get acted on $O(\sqrt{n})$ times and the input qubits $x_i$ of $U_\Phi$ get acted on only $O(\operatorname{occ}(x_i))$ times, where $\operatorname{occ}(x_i)$ denotes the number of clauses of $\Phi$ that $x_i$ appears in. 

The only issue is that the $O(\log n)$ ancilla qubits from the circuit $U_\Phi$ get acted upon a linear number of times. We now redesign the counter construction so all ancilla qubits also have degree at most $O(\sqrt{n})$. We also make sure the runtime of the verification algorithm stays linear.

\begin{theorem}[Bounded-degree \(k\)-SAT verification circuit]
\label{thm:two-unary-sat-verifier}
Fix positive integers \(k\) and \(\Delta\).
Let
\[
    \Phi(x_1,\ldots,x_n)
    =
    C_1\wedge C_2\wedge\cdots\wedge C_m
\]
be a \(k\)-CNF formula. Assume that 
\[
    \operatorname{occ}_{\Phi}(x_j)\le \Delta
    \qquad
    \text{for every }j\in[n],
\]
where \(\operatorname{occ}_{\Phi}(x_j)\) counts the number of (positive or negative) occurrences of variable $x_j$. This implies $m \le \Delta n$

\[
    B:=\left\lceil\sqrt{m+1}\right\rceil = O_\Delta(\sqrt{n}),
    \qquad
    L:=\left\lceil\sqrt m\right\rceil = O_\Delta(\sqrt{n}),
    \qquad
    R:=\left\lceil\frac{m}{L}\right\rceil = O_\Delta(\sqrt{n}).
\]
Then there exists a polynomial-time constructible quantum circuit \(U_\Phi\), acting on an \(n\)-qubit input register \(\mathsf{in}\), two \(B\)-qubit one-hot unary registers \(\mathsf{row}\) and \(\mathsf{col}\), \(kR\) clean work qubits in a register \(\mathsf{wrk}\) and one output qubit \(\mathsf{out}\). The circuit is such that, for every \(x\in\{0,1\}^n\),
\begin{equation}
\label{eq:two-unary-verifier-action}
\begin{split}
    U_\Phi
    \lvert x\rangle_{\mathsf{in}}
    \ket{0}\reg{row \cup col \cup wrk \cup out}
    &=
    \lvert x\rangle_{\mathsf{in}}
    \ket{\psi_x}\reg{row \cup col \cup wrk}
    \lvert\Phi(x)\rangle_{\mathsf{out}},
\end{split}
\end{equation}
for some state $\lvert\psi_x\rangle$ depending on $x$.

Moreover, $U_\Phi$ consists of $O_{k,\Delta}(n)$ elementary gates and the number of gates $d_{U_\Phi}(q)$ acting non-trivially on qubit $q$ satisfy
\[
    d_{U_\Phi}(x_j)
    =
    O_k\bigl(\operatorname{occ}_{\Phi}(x_j)\bigr)
\]
for every input qubit \(x_j\),
\[
    d_{U_\Phi}(a)=O_k(L) = O_{k,\Delta}(\sqrt{n})
\]
for every work qubit \(a\in\mathsf{wrk}\),
\[
    d_{U_\Phi}(q)=O_k(B) = O_{k,\Delta}(\sqrt{n})
\]
for every qubit
\(q\in\mathsf{row}\cup\mathsf{col}\), and
\[
    d_{U_\Phi}(\mathsf{out})=O(1),
\]
for the output qubit.
\end{theorem}

\begin{proof}
\noindent\textbf{Snaking counter construction.}
We first construct a sequence of \(m+1\) distinct counter states in \([B]\times[B]\), arranged in a snaking pattern. For every
\[
    t\in\{0,1,\ldots,B^2-1\} \quad \text{(note $B^2 \ge m+1$)},
\]
define
\[
    r_t
    :=
    1+\left\lfloor\frac{t}{B}\right\rfloor
\]
and
\[
    s_t
    :=
    t-B\left\lfloor\frac{t}{B}\right\rfloor
    \in\{0,\ldots,B-1\}.
\]
Set
\[
    c_t
    :=
    \begin{cases}
        s_t+1,
            & r_t\text{ is odd},\\[1mm]
        B-s_t,
            & r_t\text{ is even}.
    \end{cases}
\]
Thus, the sequence $(r_t,c_t)$ goes as
\[
\begin{split}
    &(1,1),(1,2),\ldots,(1,B), (2,B),(2,B-1),\ldots,(2,1), (3,1),(3,2),\ldots,(3,B), (4,B), (4,B-1) \dots
\end{split}
\]
For \((r,c)\in[B]\times[B]\), define the encoded counter state
\[
    \lvert r,c\rangle_{\mathsf{cnt}}
    :=
    \lvert e_r\rangle_{\mathsf{row}}
    \lvert e_c\rangle_{\mathsf{col}}.
\]
where \(\lvert e_j\rangle\) denotes the \(B\)-bit string having a single \(1\) in position \(j\).
Let
\[
    \mathcal{C}_B
    :=
    \operatorname{span}
    \{
        \lvert r,c\rangle:
        (r,c)\in[B]\times[B]
    \}
\]
be the subspace in which both registers have Hamming weight one.

\par\medskip
\noindent\textbf{Incrementing the counter.}
Fix \(i\in[m]\).
There are two possible types of transition to increment the counter.

If
\[
    r_{i-1}=r_i=:r
    \quad\text{and}\quad
    c_{i-1}\ne c_i,
\]
define
\[
    J_i
    :=
    \operatorname{CSWAP}_{
        \mathsf{row}[r];
        \mathsf{col}[c_{i-1}],
        \mathsf{col}[c_i]
    }.
\]
That is, the row qubit \(\mathsf{row}[r]\) controls a swap between the two relevant column qubits.

If
\[
    c_{i-1}=c_i=:c
    \quad\text{and}\quad
    r_{i-1}\ne r_i,
\]
define
\[
    J_i
    :=
    \operatorname{CSWAP}_{
        \mathsf{col}[c];
        \mathsf{row}[r_{i-1}],
        \mathsf{row}[r_i]
    }.
\]

In either case, let \(g_i\) denote the control qubit and let \(u_i,v_i\) denote the two swap targets, so that
\[
    J_i=\operatorname{CSWAP}_{g_i;u_i,v_i}.
\]
Note that
\begin{equation}
\label{eq:two-unary-transposition}
    J_i\lvert r,c\rangle
    =
    \begin{cases}
        \lvert r_i,c_i\rangle,
            & (r,c)=(r_{i-1},c_{i-1}),\\[1mm]
        \lvert r_{i-1},c_{i-1}\rangle,
            & (r,c)=(r_i,c_i),\\[1mm]
        \lvert r,c\rangle,
            & \text{otherwise}.
    \end{cases}
\end{equation}
Thus, one may apply \(J_i\) without first checking whether the counter is in the state \(\lvert r_{i-1},c_{i-1}\rangle\), since \(J_i\) exchanges only the two states
\[
\lvert r_{i-1},c_{i-1}\rangle
\quad\text{and}\quad
\lvert r_i,c_i\rangle
\]
and fixes every other state in $\cC_B$.

\par\medskip
\noindent\textbf{Checking the clauses}
Next, let \(\widehat J_i\) be the operator that applies \(J_i\) precisely when the clause \(C_i\) is false. More explicitly, for every computational basis assignment \(x\),
\begin{equation}
\label{eq:false-controlled-transition}
    \widehat J_i
    \lvert x\rangle_{\mathsf{in}}
    \lvert\psi\rangle_{\mathsf{cnt}}
    =
    \begin{cases}
        \lvert x\rangle_{\mathsf{in}}
        J_i\lvert\psi\rangle_{\mathsf{cnt}},
            & C_i(x)=0,\\[1mm]
        \lvert x\rangle_{\mathsf{in}}
        \lvert\psi\rangle_{\mathsf{cnt}},
            & C_i(x)=1.
    \end{cases}
\end{equation}
Note that $\widehat J_i$ can be implemented conditioning $J_i$ on the qubits representing the variables occurring in $C_i$.

Now define the \(i\)-th clause update by
\[
    V_i
    :=
    \widehat J_iJ_i.
\]
Since \(J_i^2=I\), we have
\begin{equation}
\label{eq:satisfied-clause-update}
    V_i
    \lvert x\rangle_{\mathsf{in}}
    \lvert\psi\rangle_{\mathsf{cnt}}
    =
    \begin{cases}
        \lvert x\rangle_{\mathsf{in}}
        J_i\lvert\psi\rangle_{\mathsf{cnt}},
            & C_i(x)=1,\\[1mm]
        \lvert x\rangle_{\mathsf{in}}
        \lvert\psi\rangle_{\mathsf{cnt}},
            & C_i(x)=0.
    \end{cases}
\end{equation}
Thus, the counter moves by one step when \(C_i\) is satisfied and does nothing when \(C_i\) is not satisfied.

\par\medskip
\noindent\textbf{Implementation over elementary gates.}
The $J_i$ and $\widehat J_i$ operators are multi-controlled SWAP operators with at most $k + 1$ controls. These can indeed be implemented using $O_k(1)$ standard gates and $k$ ancilla qubits in $\mathsf{wrk}$. The ancillas return to $\ket{0}$, so the ancillas may be reused. By partitioning the clauses into $R$ subsets and using fresh ancillas for each subset we can ensure that there are $kR$ ancilla qubits that are all acted upon at most $O_k(L)$ times.

\par\medskip
\noindent\textbf{Full verification circuit.}
The full verification circuit will now be given by
\[
    U_\Phi
    :=
    Q_mV_mV_{m-1}\cdots V_1P_0,
\]
where $P_0$ is the circuit preparing $\ket{1,1}_{\mathsf{cnt}}$ in the $\mathsf{row}$ and $\mathsf{col}$ registers and $Q_m$ is the circuit flipping the output qubit to 1 if the counter state is $\ket{r_m,c_m}_{\mathsf{cnt}}$.
Note that if $\Phi$ is satisfiable, every $V_i$ will increment the counter, whereas if $\Phi$ is unsatisfiable, it will fail to increment at least once, and once it has done so it will no longer increment at all.
We thus have
\begin{equation}
\label{eq:two-unary-stopping}
\begin{split}
    U_\Phi
    \lvert x\rangle_{\mathsf{in}}
    \lvert0^B\rangle_{\mathsf{row}}
    \lvert0^B\rangle_{\mathsf{col}}
    \lvert0^{kR}\rangle_{\mathsf{wrk}}\ket{0}_{\mathsf{out}}
    &=
    \lvert x\rangle_{\mathsf{in}}
    \lvert r_{\tau(x)},c_{\tau(x)}\rangle_{\mathsf{cnt}}
    \lvert0^{kR}\rangle_{\mathsf{wrk}}\ket{\Phi(x)}_{\mathsf{out}},
\end{split}
\end{equation}
where
\[
    \tau(x)
    :=
    \begin{cases}
        m,
            & \Phi(x)=1,\\[1mm]
        \min\{i\in[m]:C_i(x)=0\}-1,
            & \Phi(x)=0.
    \end{cases}
\]

Each \(V_i\) is implemented using \(O_k(1)\) elementary gates.
The preparation \(P_0\) uses two elementary gates, and the final readout \(Q_m\) uses \(O(1)\) elementary gates. Consequently,
\[
    |U_\Phi|=O_k(m).
\]
Furthermore, the number of ancilla qubits satisfy the claimed bounds by construction.

\par\medskip
\noindent\textbf{Gate-degree bounds}
As for the gate-degree bounds, because \(k\) is fixed, the elementary-gate implementation of any \(V_i\) contains at most $O_k(1)$ gates, so any $V_i$ acts on each participating qubits at with at most $O_k(1)$ elementary gates.
An input qubit \(x_j\) participates only in updates corresponding to clauses in which \(x_j\) or \(\neg x_j\) occurs. Therefore,
\[
    d_{U_\Phi}(x_j)
    \le
    O_k\bigl(\operatorname{occ}_{\Phi}(x_j)\bigr).
\]

Because we use new work qubits after processing $L$ clauses, every work qubit is acted on only $O_k(L) = O_k(\sqrt{m})$ times.

It remains to bound the degrees of the two unary counter registers.
Fix a row qubit \(\mathsf{row}[r]\). It participates as the control qubit in the \(B-1\) horizontal transitions along row \(r\), and as a swap target in at most two vertical transitions connecting row \(r\) to its neighboring rows. Hence it belongs to at most
\[
    B+1
\]
of the transition triples \(\{g_i,u_i,v_i\}\).

Now fix a column qubit \(\mathsf{col}[c]\). In each row it participates as a swap target in at most two horizontal transitions, for at most \(2B\) horizontal transitions in total. It can additionally serve as the control of a vertical transition. There are at most \(B-1\) vertical transitions altogether. Thus it belongs to fewer than
\[
    3B
\]
transition triples.

The circuit uses only a prefix of the full snake path, so these bounds also hold for the actual sequence of \(m\) transitions. Consequently, every counter qubit \(q\in\mathsf{row}\cup\mathsf{col}\) satisfies
\[
    d_{U_\Phi}(q)
    \le
    3\gamma_kB+O(1)
    =
    O_k(B).
\]
The \(O(1)\) term accounts for the initial preparation and the final readout.

Finally, the output qubit participates only in the elementary-gate decomposition of \(Q\), and hence
\[
    d_{U_\Phi}(\mathsf{out})=O(1).
\]

Since the total number of literal occurrences is at most \(\Delta n\),
\[
    m\le\Delta n.
\]
It follows that
\[
    2B+kR+1
    =
    O_{k,\Delta}(\sqrt n)
\]
and
\[
    \max_q d_{U_\Phi}(q)
    =
    O_{k,\Delta}(\sqrt n).
\]
This completes the proof.
\end{proof}

Using this new verification circuit, we obtain the following (Q)SETH hardness result, (essentially) matching \cref{cor:smallnumberhighdegree}.

\begin{corollary}
\label{cor:hardnesssmallhigh}
    The (Q)SETH hardness of quantum 5-SAT remains even when restricting to quantum SAT instances with $N$ qubits, $m = O(N)$ clauses and a partition $L \cup H = [N]$ of low and high degree qubits where $|H| = O(\sqrt{N})$ and 
    \begin{align}
        q \in L &\implies d_q = O(1) \\
        q \in H &\implies d_q = O(\sqrt{N}).
    \end{align}
    Here $d_q$ denotes the number of constraints acting on qubit $q$. 
\end{corollary}
\begin{proof}
    Assume an $\myO{2^{(1 - \varepsilon)n}}$ algorithm exists for the degree bounded version of quantum 5-SAT from the corollary statement. We will now contradict SETH. Let $\Phi$ be a $k$-SAT instance. Choose $\eta < \varepsilon$ and use the Sparsification Lemma to write $\Phi = \bigvee_i \Phi_i$ as a disjunction over $2^{\eta n}$ formulae. The $\Phi_i$ will have $O_{k, \eta}(n)$ clauses and every variable occurs at most $O_{k,\eta}(1)$ times.
    
    We now solve the $\Phi_i$ along the lines of \cref{thm:sethhard}, but using the modified verification circuit from \cref{thm:two-unary-sat-verifier}. This circuit acts on $n$ input qubits, $O(\sqrt{n})$ ancilla qubits, and consists of $T = O_k(n)$ elementary gates.

    The Hamiltonian resulting from \cref{thm:compressedcirctoQSAT} will act on a clock register of $4\sqrt{T} = O_k(\sqrt{n})$ qubits, and on a computational register consisting of $n$ input qubits and $O(\sqrt{n})$ ancilla qubits. The total number of qubits is thus $N = n + O_k(\sqrt{n})$. 

    The number of clauses of $H$ will be $O_k(N)$, as there will be $O(T) = O_k(N)$ incrementing projectors and $\Hcomp_{KL}$ consists of $O(T) = O_k({N})$ clauses, $\Hin$ of $O(\sqrt{N})$ clauses (this is the size of the ancilla register) and $\Hout$ of a single clause.

    As for the qubit degrees, by the Sparsification Lemma, every variable appears in $O_{k, \eta}(1)$ clauses. This means the verification circuit from \cref{thm:two-unary-sat-verifier} acts on each input qubit at most $O_{k,\eta}(1)$ times, and hence the input qubits are acted upon by at most $O_{k, \eta}(1)$ constraints of $H$ each.

    The clock qubits are acted upon at most $O(\sqrt{T}) = O_{k,\eta}(\sqrt{N})$ times each. Finally, by \cref{thm:two-unary-sat-verifier} the ancilla qubits are acted upon $O_k(\sqrt{N})$ times.

    There are thus $O_k(\sqrt{N})$ qubits that are acted upon by $O_k(\sqrt{N})$ projectors (the clock and ancilla qubits), whereas the other qubits (the input qubits) are acted upon only $O_{k, \eta}(1)$ times. This means that the hypothetical quantum $5$-SAT algorithm can be used to solve the formulae $\Phi_i$ in time $\myO{2^{(1 - \varepsilon')n}}$ for any $0 < \varepsilon' < \varepsilon$. Choosing $\eta < \varepsilon' < \varepsilon$, there is thus a $2^{\eta n} \myO{2^{(1 - \varepsilon') n}} = \myO{2^{(1 - (\varepsilon' - \eta))n}}$ algorithm for $k$-SAT. As $\varepsilon' - \eta > 0$, this contradicts SETH. The proof for QSETH is similar except we need to set $\eta < \frac{\varepsilon}{2}$. 
\end{proof}

\section*{AI disclosure}
All original ideas were derived from the authors, and ChatGPT 5.5 and 5.6 were used to aid in the write-up. The authors take full responsibility for all content of this manuscript.

\section*{Acknowledgments}
AH thanks Thomas Vidick for discussions.

AH is supported by JSPS KAKENHI grant No.~24H00071, 25K24674, 25K2446. JK was supported by the DFG under grant number 450041824. FLG is supported by JSPS KAKENHI grant No.~24H00071, 25K24674, 25K2446, MEXT Q-LEAP grant No.~JPMXS0120319794, JST ASPIRE grant No.~JPMJAP2302 and JST CREST grant No.~JPMJCR24I4. ST is supported by JSPS KAKENHI grant No.~JP22K11909. 

\bibliographystyle{alpha}
\bibliography{ref}

\appendix

\section{Proof of \cref{lem:U_compute_Phi}}\label{appendix}

\begin{proof}
All logarithms are base two.  Set
\[
    r:=\left\lceil\log_2(m+1)\right\rceil,
\]
so that every integer in $\{0,1,\ldots,m\}$ has an $r$-bit binary representation.  
The circuit uses an $r$-qubit counter register $\mathsf{cnt}$, a 1-qubit clause register $\mathsf{cls}$, a 1-qubit output register $\mathsf{out}$, and $O(r)$ clean work qubits that are reused from one step to the next.

\paragraph{Computing one clause.}
For each $i\in[m]$, let $W_i$ be a reversible circuit satisfying
\[
    W_i\lvert x\rangle_{\mathrm{in}}\lvert b\rangle_{\mathsf{cls}}
    =
    \lvert x\rangle_{\mathrm{in}}
    \lvert b\oplus C_i(x)\rangle_{\mathsf{cls}}
\]
for all $x\in\{0,1\}^n$ and $b\in \{0,1\}$.  Since $C_i$ is the OR of at most $k$ literals, $W_i$ can be implemented by a multi-controlled NOT, together with $O(k)$ NOT gates. Using clean work qubits, both $W_i$ and $W_i^\dagger$ therefore have size $O(k)$ and use $O(k)$ reusable work qubits.

\paragraph{The step-dependent counter update.}
For a positive integer $i$, let $\nu_2(i)$ be the largest integer $s\geq 0$ such that $2^s$ divides $i$, and define
\[
    t_i:=1+\nu_2(i).
\]
Equivalently, $t_i$ is the number of bits that change when the binary integer $i-1$ is incremented to $i$.  More explicitly, for some possibly empty prefix $p_i$,
\[
    i-1=(p_i\,0\,1^{t_i-1})_2,
    \qquad
    i=(p_i\,1\,0^{t_i-1})_2.
\]

For $1\leq t\leq r$, define $\Suffix_t$ on the computational basis of the counter by
\begin{align*}
    \Suffix_t\lvert p\,0\,1^{t-1}\rangle
        &=\lvert p\,1\,0^{t-1}\rangle,\\
    \Suffix_t\lvert p\,1\,0^{t-1}\rangle
        &=\lvert p\,0\,1^{t-1}\rangle
\end{align*}
for every prefix $p\in\{0,1\}^{r-t}$, while fixing all other basis states.  Thus $\Suffix_t$ swaps two suffix patterns and is an involution.  In particular, it is unitary.  It has the following two properties:
\begin{enumerate}
    \item $\Suffix_{t_i}$ maps the ideal counter value $i-1$ to $i$;
    \item for every integer $z$ represented by the counter,
    \[
        \Suffix_t(z)\in\{z-1,z,z+1\}.
    \]
\end{enumerate}
The second property holds because, for each fixed prefix $p$, the two swapped binary strings represent adjacent integers.

For example, the first few updates are
\[
    0\leftrightarrow1,
    \qquad
    01\leftrightarrow10,
    \qquad
    0\leftrightarrow1,
    \qquad
    011\leftrightarrow100,
    \qquad\ldots
\]
where each operation acts only on the displayed least significant bits. 
Let $\CtrlSuffix{t}$ denote $\Suffix_t$ controlled by the clause qubit $\cls$.

\begin{claim}
\label{claim:suffix-implementation}
The controlled update $\CtrlSuffix{t}$ can be implemented using $O(t)$ elementary gates and $O(t)$ clean work qubits.
\end{claim}

\begin{proof}
For $t=1$, the operation is simply a $\CNOT$ from $\cls$ to the least significant counter bit.
Assume $t\geq2$, and write the $t$ least significant counter bits as $q_{t-1}\cdots q_1q_0$.

Introduce clean qubits $d_0,\ldots,d_{t-2}$ and a flag qubit $a$.  First compute
\[
    d_j=q_j\oplus q_{t-1}
    \qquad (0\leq j\leq t-2).
\]
Then
\[
    d_0=d_1=\cdots=d_{t-2}=1
\]
if and only if the suffix is either $0\,1^{t-1}$ or $1\,0^{t-1}$. Compute the flag
\[
    a=\cls\wedge d_0\wedge\cdots\wedge d_{t-2}
\]
with a multi-controlled NOT.  Controlled by $a$, complement all $t$ suffix bits.  This exchanges $0\,1^{t-1}$ and $1\,0^{t-1}$.  Apply the same multi-controlled NOT again to reset $a$, and then uncompute the $d_j$'s.

The uncomputation is valid because complementing both $q_j$ and $q_{t-1}$ leaves $q_j\oplus q_{t-1}$ unchanged.  The construction uses $O(t)$ CNOT gates and two multi-controlled NOT gates with $t$ controls.  Each such gate has an exact $O(t)$-size implementation using $O(t)$ clean ancillas. Hence the total size and work-space costs are both $O(t)$.
\end{proof}

\paragraph{The complete verification procedure.}
Initialize $\cnt$ to $0$.  For each $i=1,\ldots,m$, compute $C_i(x)$ into $\cls$, apply the controlled update $\CtrlSuffix{t_i}$, and uncompute $\cls$.  Finally, apply a circuit $\Compare_m$ that flips $\outreg$ if and only if the counter contains $m$. The gate \(\Compare_m\) is implemented by conjugating an \(r\)-controlled NOT by NOT gates on precisely those counter bits where the binary expansion of \(m\) has value \(0\). Hence \(\Compare_m\) has exact \(O(r)\) size and uses \(O(r)\) clean work qubits. Thus
\[
    U_\Phi
    :=
    \Compare_m
    \bigl(W_m^\dagger \CtrlSuffix{t_m}W_m\bigr)
    \cdots
    \bigl(W_1^\dagger \CtrlSuffix{t_1}W_1\bigr).
\]
The comparison with the fixed string $\operatorname{bin}(m)$ uses $O(r)$ elementary gates and $O(r)$ reusable work qubits.

\paragraph{Completeness.}
Fix $x\in\bits^n$.
Let $z_i$ be the integer in the counter after the first $i$ clauses have been processed, with $z_0=0$. Suppose first that every clause is satisfied.  Then $C_i(x)=1$ for every $i$.  If $z_{i-1}=i-1$, the definition of $t_i$ gives
\[
    \Suffix_{t_i}(i-1)=i.
\]
Induction therefore gives $z_i=i$ for all $i$, and in particular $z_m=m$. Hence the output bit is $1$.

\paragraph{Soundness: the deficit never decreases.}
Define the deficit after step $i$ by
\[
    \Delta_i:=i-z_i.
\]
The key point is that a controlled suffix update can increase the counter by at most one. If $C_i(x)=1$, then
\[
    z_i\leq z_{i-1}+1,
\]
and consequently
\[
    \Delta_i=i-z_i
    \geq i-(z_{i-1}+1)
    =(i-1)-z_{i-1}
    =\Delta_{i-1}.
\]
If $C_i(x)=0$, the update is skipped, so $z_i=z_{i-1}$ and
\[
    \Delta_i=i-z_{i-1}=\Delta_{i-1}+1.
\]
Thus the deficit never decreases, and every unsatisfied clause increases it by one. Since $\Delta_0=0$, if at least one clause evaluates to zero, then
\[
    \Delta_m\geq1,
\]
and therefore
\[
    z_m=m-\Delta_m\leq m-1.
\]
Hence the final counter equals $m$ if and only if all clauses are satisfied.  The equality test therefore writes precisely $\Phi(x)$ into $\outreg$.

\paragraph{Gate and ancilla complexity.}
Computing and uncomputing the clauses uses $O(km)$ elementary gates. By \cref{claim:suffix-implementation}, the counter updates use
\[
    O\left(\sum_{i=1}^m t_i\right)
\]
gates. Moreover,
\begin{align*}
    \sum_{i=1}^m t_i
    &=m+\sum_{i=1}^m\nu_2(i)\\
    &=m+\sum_{s\geq1}\left\lfloor\frac{m}{2^s}\right\rfloor\\
    &<2m.
\end{align*}
Thus all counter updates together have size $O(m)$, even though an individual update may touch $\Theta(\log m)$ bits. The final $\Compare_m$ gate has size $O(r)$.  The total number of gates is therefore
\[
    O(km+m+r)=O\bigl(km+\log(m+1)\bigr),
\]
where we used $k\geq1$. Since all work qubits can be reused, the total number of ancilla qubits is $O(k+r)=O(k+\log(m+1))$. Since $k$ is constant and $m=O(n)$, the number of gates is $O(n)$ and the size of ancilla qubits is $O(\log n)$.
The circuit description is computable in time polynomial in the size of $\Phi$.
\end{proof}

\end{document}